\documentclass{article}

     \PassOptionsToPackage{numbers, compress}{natbib}

     \usepackage[preprint]{neurips_2019}

\usepackage[utf8]{inputenc} 
\usepackage[T1]{fontenc}    
\usepackage{hyperref}       
\usepackage{url}            
\usepackage{booktabs}       
\usepackage{amsfonts}       
\usepackage{nicefrac}       
\usepackage{microtype}      

\usepackage{thm-restate}

\usepackage{amsmath}\allowdisplaybreaks
\usepackage{amssymb}
\usepackage{amsthm}
\usepackage{dsfont}
\usepackage{mathrsfs}
\usepackage[all]{xy}
\usepackage{empheq}
\usepackage{algorithm,algorithmic}
\usepackage{multirow}
\usepackage{multicol}
\usepackage{ifthen}
\usepackage[usenames]{color}
\usepackage[usenames,dvipsnames]{xcolor}
\usepackage{epstopdf}
\usepackage{verbatim}
\usepackage{fancybox,framed}
\usepackage{array}

\usepackage{hyperref}
\hypersetup{
	colorlinks   = true, 
	urlcolor     = blue, 
	linkcolor    = blue, 
	citecolor   = blue 
}

\title{Adaptive Influence Maximization with \\ Myopic Feedback}

\author{%
  Binghui Peng \\
  Tsinghua University\\
  \texttt{pbh15@mails.tsinghua.edu.cn}\\
  \And
  Wei Chen\\
  Microsoft Research\\
  \texttt{weic@microsoft.com}\\
}

\newtheorem{theorem}{Theorem}
\newtheorem{lemma}{Lemma}

\newtheorem{claim}{Claim}
\newtheorem{definition}{Definition}

\newcommand{\E}{\mathop{\mathbb{E}}}
\newcommand{\I}{\mathbb{I}}
\newcommand{\R}{\mathbb{R}}

\newcommand{\cA}{\mathcal{A}}

\newcommand{\cP}{\mathcal{P}}
\newcommand{\cR}{\mathcal{R}}

\newcommand{\cW}{\mathcal{W}}
\newcommand{\argmax}{\mathrm{argmax}}

\newcommand{\dom}{{\rm dom}}

\newcommand{\G}{{\rm Greedy}}

\newcommand{\fa}{\sigma}
\newcommand{\bfa}{\bar{\sigma}}
\newcommand{\OPT}{\mathrm{OPT}}

\newcommand{\Evt}{\mathcal{E}}

\newcommand{\T}{\mathcal{T}}
\newcommand{\aL}{\mathcal{L}}
\newcommand{\Sps}{S^{+}}
\newcommand{\bp}{\bar{\pi}}

\begin{document}

\maketitle

\begin{abstract}
We study the {\em adaptive influence maximization problem with myopic feedback} under the
	independent cascade model: one sequentially selects $k$ nodes as seeds one by one from a social network,
	and each selected seed returns the immediate neighbors it activates as the feedback available
	for later selections, and the goal is to maximize the expected number of total activated nodes,
	referred as the {\em influence spread}.
We show that the {\em adaptivity gap}, the ratio between the optimal adaptive influence spread and
	the optimal non-adaptive influence spread, is at most $4$ and at least $e/(e-1)$, 
	and the approximation ratios with respect to the optimal adaptive influence spread
	of both the non-adaptive greedy and adaptive greedy algorithms are at least 
	$\frac{1}{4}(1 - \frac{1}{e})$ and at most $\frac{e^2 + 1}{(e + 1)^2} < 1 - \frac{1}{e}$.
Moreover, the approximation ratio of the non-adaptive greedy algorithm is no worse than
	that of the adaptive greedy algorithm, when considering all graphs.
Our result confirms a long-standing open conjecture of 
	Golovin and Krause (2011) on the constant 
	approximation ratio of adaptive greedy with myopic feedback, and
	it also suggests that adaptive greedy may not bring much benefit under myopic feedback.

\end{abstract}
\section{Introduction}

Influence maximization is the task of given a social network and a stochastic 
diffusion model on the network, 
finding the $k$ seed nodes with the largest expected influence spread in the
model~\cite{Kempe2003maximizing}.
Influence maximization and its variants have applications in viral marketing, rumor control, etc. and
have been extensively studied (cf. \cite{chen2013information,LiFWT18}).

In this paper, we focus on the {\em adaptive influence maximization} problem, where seed nodes are
sequentially selected one by one, and after each seed selection, partial or full
	 diffusion results from the seed are returned
as the feedback, which could be used for subsequent seed selections.
Two main types of feedback has been proposed and studied before:
	(a) {\em full-adoption feedback}, where the entire diffusion process from the seed selected is 
	returned as the feedback, and
	(b) {\em myopic feedback}, where only the immediate neighbors activated by
	the selected seed are returned as the feedback.
Under the common independent cascade (IC) model where every edge in the graph has an 
	independent probability of passing influence, \citet{GoloKrause11} show that 
	the full-adoption feedback model satisfies the key {\em adaptive submodularity} property, 
	which enables a simple adaptive greedy algorithm to achieve a $(1-1/e)$ approximation 
	to the adaptive optimal solution.
However, the IC model with myopic feedback is not adaptive submodular, and 
\citet{GoloKrause11} only conjecture that in this case the adaptive greedy algorithm still
guarantees a constant approximation.
To the best of our knowledge, this conjecture is still open before our result in this paper,
which confirms that indeed adaptive greedy is a constant approximation of the adaptive optimal
solution.

In particular, our paper presents two sets of related results on adaptive influence maximization
with myopic feedback under the IC model.
We first study the {\em adaptivity gap} of the problem (Section~\ref{sec:adaptivity-gap-myopic}),
which is defined as the ratio between the
adaptive optimal solution and the non-adaptive optimal solution, and is an indicator on how
useful the adaptivity could be to the problem.
We show that the adaptivity gap for our problem is at most $4$ 
(Theorem~\ref{theorem:adaptive-gap-upper}) and at least $e/(e-1)$
(Theorem~\ref{theorem:adaptive-gap-lower}).
The proof of the upper bound $4$ is the most involved, because the problem is not adaptive submodular,
and we have to create a hybrid policy that involves three independent runs of the diffusion
process in order to connect between an adaptive policy and a non-adaptive policy.
Next we study the approximation ratio with respect to the adaptive optimal solution
for both the non-adaptive greedy and adaptive greedy algorithms (Section~\ref{sec:constant-competitive-algo}).
We show that the approximation ratios of both algorithms are at least $\frac{1}{4}(1 - \frac{1}{e})$
(Theorem~\ref{theorem:greedy&adaptive-greedy-constant}), which
combines the adaptivity upper bound of $4$ with the results that both algorithms achieve
$(1-1/e)$ approximation of the non-adaptive optimal solution (the $(1-1/e)$ approximation
ratio for the adaptive greedy algorithm requires a new proof).
We further show that the approximation ratios for both algorithms are at most 
$\frac{e^2 + 1}{(e + 1)^2}\approx 0.606$, which is strictly less than $1-1/e \approx 0.632$, and
the approximation ratio of adaptive greedy is at most that of the non-adaptive greedy
when considering all graphs (Theorem~\ref{theorem:bad-example-greedy}).

In summary, our contribution is the systematic study on adaptive influence maximization with
myopic feedback under the IC model.
We prove both constant upper and lower bounds on the adaptivity gap in this case, and 
constant upper and lower bounds on the approximation ratios 
(with respect to the optimal adaptive solution) achieved by 
non-adaptive greedy and adaptive greedy algorithms.
The constant approximation ratio of the adaptive greedy algorithm answers a long-standing open
conjecture affirmatively.
Our result on the adaptivity gap is the first one on a problem not satisfying adaptive submodularity.
Our results also suggest that adaptive greedy may not bring much benefit under the myopic feedback
model.

Due to the space constraint, full proof details are included in the supplementary material.

\vspace{2mm}
\noindent{\bf Related Work.\ \ }
Influence maximization as a discrete optimization task is first proposed by \citet{Kempe2003maximizing},
	who propose the independent cascade, linear threshold and other models, study their submodularity
	and the greedy approximation algorithm for the influence maximization task.
Since then, influence maximization and its variants have been extensively studied. 
We refer to recent surveys~\cite{chen2013information,LiFWT18} for the general coverage of this
	area.

Adaptive submodularity is formulated by \citet{GoloKrause11} for general stochastic adaptive
	optimization problems, and they show that the adaptive greedy algorithm achieves
	$1-1/e$ approximation if the problem is adaptive monotone and adaptive submodular.
They study the influence maximization problem under the
	IC model as an application, and prove that the full-adoption feedback under the IC model is adaptive
	submodular.
However, in their arXiv version, they show that the myopic feedback version is not adaptive
	submodular, and they conjecture that adaptive greedy would still achieve a constant 
	approximation in this case.
 

Adaptive influence maximization has been studied 
	in~\cite{tong2017adaptive,yuan2017adaptive,sun2018multi,salha2018adaptive,tong2019adaptive, pmlr-v97-fujii19a}.
\citet{tong2017adaptive} provide both adaptive greedy and efficient heuristic algorithms for
	adaptive influence maximization.
Their theoretical analysis works for the full-adoption feedback model but has a gap when applied
	to myopic feedback, which is confirmed by the authors.
\citet{yuan2017adaptive} introduce the partial feedback model and develop algorithms that balance
	the tradeoff between delay and performance, and their partial feedback model does not 
	coincide with the myopic feedback model.
\citet{salha2018adaptive} consider a different diffusion model where edges can be reactivated
	at each time step, and they show that myopic feedback under this model 
	is adaptive submodular.
\citet{sun2018multi} study the multi-round adaptive influence maximization problem, where 
	$k$ seeds are selected in each round and at the end of the round the full-adoption feedback is 
	returned.
\citet{tong2019adaptive} introduces a general feedback model and develops some heuristic algorithms
	for this model.
\citet{pmlr-v97-fujii19a} study the adaptive influence maximization problem on \emph{bipartite} graphs. 
They introduce the notion of adaptive submodularity ratios and prove constant approximation ratio for the adaptive greedy algorithm. 
However, as noted by the authors, their results are not applicable to the general graph case and the adaptive submodularity ratio can be arbitrarily small in general graphs.
	
A different two stage seeding process has also been
	studied~\cite{seeman2013adaptive,badanidiyuru2016locally,singer2016influence}, but the model
	is quite different, since their first stage of selecting a node set $X$ is only to
	introduce the neighbors of $X$ as seeding candidates for the second stage.


Adaptivity gap has been studied by two lines of research.
The first line of work utilizes multilinear extension and adaptive submodularity to study
	adaptivity gaps for the class of stochastic submodular maximization problems and give a $e/(e - 1)$ upper bound for matroid constraints \cite{asadpour2008stochastic, asadpour2015maximizing}.
The second line of work \cite{gupta2016algorithms, gupta2017adaptivity, bradac2019near} studies the stochastic probing problem and proposes the idea of random-walk non-adaptive policy on the
	decision tree, which partially inspires our analysis.
However, their analysis also implicitly depends on adaptive submodularity.
In contrast, our result on the adaptivity gap is the first on a problem that does not satisfy
	adaptive submodularity (see Section~\ref{sec:adaptivityUpperBound} for more discussions).

\section{Model and Problem Definition}
\label{sec:model}

\noindent{\bf Diffusion Model.\ \ }
In this paper, we focus on the well known Independent Cascade (IC) model as the diffusion model. 
In the IC model, the social network is described by a directed influence graph 
	$G=(V, E, p)$, where $V$ is the set of nodes ($|V| = n$), 
	$E \subseteq V\times V$ is the set of directed edges, and
	each directed edge $(u, v) \in E$ is associated with a probability $p_{uv}\in [0, 1]$. 
The \emph{live edge} graph $L = (V, L(E))$ is a random subgraph of $G$, for any edge $(u, v) \in E$, $(u, v) \in L(E)$ with independent probability $p_{uv}$. 
If $(u, v)\in L(E)$, we say edge $(u, v)$ is \emph{live}, otherwise we say it is \emph{blocked}. 
The dynamic diffusion in the IC model is as follows:
	at time $t=0$ a live-edge graph $L$ is sampled and nodes in a seed set $S\subseteq V$ 
	 are activated.
At every discrete time $t = 1, 2, \ldots$, if a node $u$ was activated at time $t-1$, then all
	of $u$'s out-going neighbors in $L$ are activated at time $t$.
The propagation continues until there are no more activated nodes at a time step.
The dynamic model can be viewed equivalently as every activated node $u$ has one chance to activate
	each of its out-going neighbor $v$ with independent success probability $p_{uv}$.
Given a seed set $S$, the {\em influence spread} of $S$, denoted $\sigma(S)$, is the expected number of
	nodes activated in the diffusion process from $S$, i.e. $\sigma(S) = \E_{L}[|\Gamma(S,L)|]$,
	where $\Gamma(S,L)$ is the set of nodes reachable from $S$ in graph $L$.


\noindent{\bf Influence Maximization Problem.\ \ }
Under the IC model, 
	we formalize the influence maximization (IM) problem in both
	non-adaptive and adaptive settings.
Influence maximization in the non-adaptive setting follows the 
	classical work of~\cite{Kempe2003maximizing}, and is defined below.
	
\begin{definition}[Non-adaptive Influence Maximization]
{\em Non-adaptive influence maximization} is the problem of given a
	directed influence graph $G=(V,E, p)$ with IC model parameters $\{p_{uv}\}_{(u,v)\in E}$ and a
	budget $k$, finding a seed set $S^*$ of at most $k$ nodes such that the influence spread
	of $S^*$, $\sigma(S^*)$, is maximized, i.e. finding 
	$S^* \in \argmax_{S\subseteq V, |S|\le k} \sigma(S)$.
\end{definition}

We formulate influence maximization in the adaptive setting 
	following the framework of \cite{GoloKrause11}. 
Let $O$ denote the set of states, which informally 
	correspond to the feedback information in the adaptive setting.
A {\em realization} $\phi$ is a function $\phi:V \rightarrow {O}$, 
	such that for $u\in V$, $\phi(u)$ represents the feedback obtained when selecting $u$
	as a seed node.
In this paper, we focus on the {\em myopic feedback} model~\cite{GoloKrause11}, which means the feedback of 
	a node $u$ only contains the status of the out-going edges of $u$
	being live or blocked.
Informally it means that after selecting a seed we can only see its one step propagation effect as the feedback.
The realization $\phi$ then determines the status of every edge in $G$, and thus  
	corresponds to a live-edge graph.
As a comparison, the {\em full-adoption feedback} model~\cite{GoloKrause11}
	is such that for each seed node $u$, the feedback contains the status of every out-going
	edge of every node $v$ that is reachable from $u$ in a live-edge graph $L$.
This means that after selecting a seed $u$, we can see the full cascade from $u$ as the feedback.
In the full-adoption feedback case, each realization $\phi$ also corresponds to a unique 
	live-edge graph.
Henceforth, we refer to $\phi$ as both a realization and a live-edge graph interchangeably.
In the remainder of this section, 
	the terminologies we introduce apply to both feedback models, unless we
	explicitly point out which feedback model we are discussing.

Let $\cR$ denote the set of all realizations.
We use $\Phi$ to denote a random realization, following the distribution $\cP$ over random live-edge
	graphs (i.e. each edge $(u,v)\in E$ has an independent probability of $p_{uv}$ to be live
	in $\Phi$).
Given a subset $S$ and a realization $\phi$, we define {\em influence utility function}
	$f: 2^V  \times \cR \rightarrow \R^{+}$ as $f(S,\phi) = |\Gamma(S, \phi)|$, where
	$\R^{+}$ is the set of non-negative real numbers.
That is, $f(S,\phi)$ is the number of nodes reachable from $S$ in realization
	(live-edge graph) $\phi$.
Then it is clear that influence spread $\sigma(S) = \E_{\Phi\sim \cP}[f(S,\Phi)]$.

%

In the \emph{adaptive influence maximization} problem, we could sequentially select nodes as seeds,
	and after selecting one seed node, we could obtain its feedback, and use the feedback to 
	guide further seed selections.
A \emph{partial realization} $\psi$ maps a subset of nodes in $V$, denoted $\dom(\psi)$ for
	domain of $\psi$,
	to their states.
 Partial realization $\psi$ represents the feedback we could obtain after nodes in
 	$\dom(\psi)$ are selected as seeds.
For convenience, we also represent $\psi$ as a relation, i.e., 
 	$\psi = \{(u, o) \in V \times {O}: u \in \dom(\psi), o = \psi(u)\} $. 
We say that a full realization $\phi$ is \emph{consistent} with a partial realization $\psi$, denoted as $\phi \sim \psi$, if $\phi(u) = \psi(u)$ for every $u\in \dom(\psi)$.

An adaptive policy $\pi$ is a mapping from partial realizations to nodes.
Given a partial realization $\psi$, $\pi(\psi)$ represents the next seed node policy $\pi$
	would select when it sees the feedback represented by $\psi$.
Under a full realization $\phi$ consistent with $\psi$, 
	after selecting $\pi(\psi)$, the policy would obtain feedback $\phi(\pi(\psi))$, and
	the partial realization would grow to $\psi' = \psi \cup \{(\pi(\psi), \phi(\pi(\psi)))\}$,
	and policy $\pi$ could pick the next seed node $\pi(\psi')$ based on
	partial realization $\psi'$.
For convenience, we only consider deterministic policies in this paper, and
	the results we derived can be easily extend to randomized policies.
Let $V(\pi, \phi)$ denote the set of nodes selected by policy $\pi$ under realization $\phi$.
For the adaptive influence maximization problem, we consider the simple cardinality constraint
	such that $|V(\pi, \phi)| \le k$, i.e. the policy only selects at most $k$ nodes.
Let $\Pi(k)$ denote the set of such policies.

The objective of an adaptive policy $\pi$ is its {\em adaptive influence spread}, which is the
	expected number of nodes that are activated under policy $\pi$.
Formally, we define the adaptive influence spread of $\pi$ as
	$\sigma(\pi) = \E_{\Phi\sim \cP}[f(V(\pi, \Phi), \Phi)]$.
The adaptive influence maximization problem is defined as follows.

\begin{definition}[Adaptive Influence Maximization]
{\em Adaptive influence maximization} is the problem of given a
directed influence graph $G=(V,E, p)$ with IC model parameters $\{p_{uv}\}_{(u,v)\in E}$ and a
budget $k$, finding an adaptive policy $\pi^*$ that selects at most $k$ seed nodes
such that the adaptive influence spread of $\pi^*$, $\sigma(\pi^*)$, is maximized, i.e. finding 
$\pi^* \in \argmax_{\pi \in \Pi(k)} \sigma(\pi)$.
\end{definition}

Note that for any fixed seed set $S$, we can create a policy $\pi_S$ that
	always selects set $S$ regardless of the feedback, which means any non-adaptive solution
	is a feasible solution for adaptive influence maximization.
Therefore, the optimal adaptive influence spread should be at least as good as the optimal
	non-adaptive influence spread, under the same budget constraint.

\vspace{2mm}
\noindent{\bf Adaptivity Gap.\ \ }
Since the adaptive policy is usually hard to design and analyze and the adaptive interaction
	process	may also be slow in practice, a fundamental question for adaptive stochastic optimization problems is whether adaptive algorithms are really superior to non-adaptive algorithms. 
The {\em adaptivity gap} measures the gap between the optimal adaptive solution and the optimal non-adaptive solution. 
More concretely, if we use $\OPT_{N}(G, k)$ (resp. $\OPT_{A}(G, k)$) to denote the influence spread of the optimal non-adaptive (resp. adaptive) solution for the IM problem in an influence graph $G$ under the IC model with seed budget $k$, then we have the following definition.
\begin{definition}[Adaptivity Gap for IM]
	\label{definition:adaptivity-gap}
	The adaptivity gap in the IC model is defined as the supremum ratio of
	the influence spread between the optimal adaptive policy and the optimal non-adaptive policy, over all possible influence graphs and seed set size $k$, i.e., 
	\begin{align}\label{eq:adaptivity-gap-definition}
	\sup_{G, k}\frac{\OPT_{A}(G, k)}{\OPT_{N}(G, k)}.
	\end{align}
\end{definition}

\vspace{2mm}
\noindent{\bf Submodularity and Adaptive Submodularity.\ \ }
Non-adaptive influence maximization is often solved via submodular function maximization technique.
A set function $f:2^V \rightarrow \R$ is {\em submodular} if for all $S\subseteq T \subseteq V$
	and all $u \in V\setminus T$, 
	$f(S\cup \{u\}) - f(S) \ge f(T\cup \{u\}) - f(T)$.
Set function $f$ is monotone if for all $S\subseteq T \subseteq V$, $f(S) \le f(T)$.
\citet{Kempe2003maximizing} show that the influence spread function $\sigma(S)$ under
	the IC model is monotone and submodular,
	and thus a simple (non-adaptive) greedy algorithm achieves a 
	$(1-\frac{1}{e})$ approximation of the optimal (non-adaptive) solution, assuming function evaluation
	$\sigma(S)$ is given by an oracle.

\citet{GoloKrause11} define {\em adaptive submodularity} for the adaptive stochastic 
	optimization framework.
In the context of adaptive influence maximization, adaptive submodularity can be defined
	as follows.
Given a utility function $f$, for any 
	partial realization $\psi$ and a node $u\not\in \dom(\psi)$, 
	we define the marginal gain of $u$ given $\psi$ as
	$\Delta_{f}(u \mid \psi) = \E_{\Phi\sim \cP}[f(\dom(\psi) \cup \{u\}, \Phi) - f(\dom(\psi), \Phi) | \Phi \sim \psi ]$, i.e. the expected marginal gain on
	influence spread when adding $u$ to the partial realization $\psi$.
A partial realization $\psi$ is a {\em sub-realization} of another partial
	realization $\psi'$ if $\psi \subseteq \psi'$ when treating both as relations.
We say that the utility function $f$ is {\em adaptive submodular} with respect to 
	$\cP$ if for any two fixed partial realizations $\psi$ and $\psi'$ such that
	$\psi\subseteq \psi'$, for any $u \not\in \dom(\psi')$, we have
	$\Delta_{f}(u \mid \psi) \ge \Delta_{f}(u \mid \psi')$, that is,
	the marginal influence spread of a node given more feedback is at most 
	its marginal influence spread given less feedback.
We say that $f$ is {\em adaptive monotone} with respect to $\cP$ if
	for any partial realization $\psi$ with $\Pr_{\Phi \sim \cP}(\Phi \sim \psi) > 0$,
	$\Delta_{f}(u \mid \psi) \ge 0$.
	
\citet{GoloKrause11} show that the influence utility function under the
	IC model with full adoption feedback is adaptive monotone and adaptive submodular, and
	thus the adaptive greedy algorithm achieves $(1-\frac{1}{e})$ approximation of the adaptive
	optimal solution.
However, they show that the influence utility function under the IC model with myopic feedback is
	not adaptive submodular.
They conjecture that the adaptive greedy policy still provides a constant approximation.
In this paper, we show that the adaptive greedy policy provides a $\frac{1}{4}(1-\frac{1}{e})$ approximation,
	and thus finally address this conjecture affirmatively.

\section{Adaptivity Gap in Myopic Feedback Model}
\label{sec:adaptivity-gap-myopic}

In this section, we analyze the adaptivity gap for influence maximization problems under the myopic feedback model and derive both upper and lower bounds. 


\subsection{Upper Bound on the Adaptivity Gap}
\label{sec:adaptivityUpperBound}

Our main result is an upper bound on the adaptivity gap for myopic feedback models, which is formally stated below.

\begin{theorem}
	\label{theorem:adaptive-gap-upper}
	Under the IC model with myopic feedback, the adaptivity gap for the influence
	maximization problem is at most 4.
\end{theorem}

\paragraph{Proof outline.} 
We now outline the main ideas and the structure of the proof of Theorem~\ref{theorem:adaptive-gap-upper}.
The main idea is to show that for each adaptive policy $\pi$, we could construct a non-adaptive
	randomized policy $\cW(\pi)$, such that the adaptive influence spread $\sigma(\pi)$ is at most
	four times the non-adaptive influence spread of $\cW(\pi)$, denoted $\sigma(\cW(\pi))$.
This would immediately imply Theorem~\ref{theorem:adaptive-gap-upper}.
The non-adaptive policy $\cW(\pi)$ is constructed by viewing adaptive policy $\pi$ as a 
	decision tree with leaves representing the final seed set selected
	(Definition~\ref{definition:decision-tree-adaptive-policy}), and $\cW(\pi)$ simply
	samples such a seed set based on the distribution of the leaves
	(Definition~\ref{definition:random-non-adaptive}).
The key to connect $\sigma(\pi)$ with $\sigma(\cW(\pi))$ is by introducing a fictitious
	hybrid policy $\bp$, such that $\sigma(\pi) \le \bfa(\bp) \le 4 \sigma(\cW(\pi))$,
	where $\bfa(\bp)$ is the {\em aggregate adaptive influence spread}
	(defined in Eqs.~\eqref{eq: utility-fake} and~\eqref{eq: utility-hybrid-strategy}).
Intuitively, $\bp$ works on three independent realizations $\Phi^1, \Phi^2, \Phi^3$, such that
	it adaptively selects seeds just as $\pi$ working on $\Phi^1$, but each selected seed
	has three independent chances to activate its out-neighbors accordingly the union of
	$\Phi^1, \Phi^2, \Phi^3$.
The inequality $\sigma(\pi) \le \bfa(\bp) $ is immediate and the main effort is on proving
	$\bfa(\bp) \le 4 \sigma(\cW(\pi))$.
	
To do so, we first introduce general notations $\sigma^t(S)$ and $\sigma^t(\pi)$ with $t=1,2,3$, 
	where  $\sigma^t(S)$ is the
	{\em $t$-th aggregate influence spread} for a seed set $S$ and
	$\sigma^t(\pi)$ is the {\em $t$-th aggregate adaptive influence spread} for an adaptive policy $\pi$,
	and they mean that all seed nodes have $t$ independent chances to activate their out-neighbors.
Obviously, $\bfa(\bp) = \sigma^3(\pi)$ and $\sigma(\cW(\pi)) = \sigma^1(\cW(\pi))$.
We then represent $\sigma^t(S)$ and $\sigma^t(\pi)$ as a summation of $k$ non-adaptive 
	marginal gains $\Delta_{f^{t}}(u \mid \dom(\psi^1))$'s and
	adaptive marginal gains  $\Delta_{f^{t}}(u \mid \psi^1)$'s, respectively
	(Definition~\ref{definition:conditional-exptected-marginal-utility-hybrid} and
	Lemma~\ref{lem:marginal}), with respect to the different levels of the decision tree.
Next, we establish the key connection between the adaptive marginal gain and
	the nonadaptive marginal gain (Lemma~\ref{lemma:hybrid-vs-non-adaptive}):
	$ \Delta_{f^{3}}(u \mid \psi^1) \leq 2\Delta_{f^{2}}(u \mid \dom(\psi^1))$.
This immediately implies  that $\sigma^3(\pi) \le 2\sigma^2(\cW(\pi))$.
Finally, we prove that the $t$-th aggregate non-adaptive influence spread $\sigma^t(S)$ is
	bounded by $t\cdot \sigma(S)$, which implies that $\sigma^2(\cW(\pi)) \le 2\sigma(\cW(\pi))$.
This concludes the proof.

We remark that our introduction of the hybrid policy $\bp$ is inspired by the analysis in
	\cite{bradac2019near}, which shows that the adaptivity gap for the 
	\emph{stochastic multi-value probing (SMP)} problem is at most $2$.
However, our analysis is more complicated than theirs and thus is novel in several aspects.
First, the SMP problem is simpler than our problem,
	with the key difference being that SMP is adaptive submodular but our problem is not.
Therefore, we cannot apply their way of inductive reasoning that implicitly relies on adaptive 
	submodularity.
Instead, we have to use our marginal gain representation and redo the bounding analysis 
	carefully based on
	the (non-adaptive) submodularity of the influence utility function on live-edge graphs.
Moreover, our influence utility function is also sophisticated and we have to use three 
	independent realizations in order to apply the submodularity on live-edge graphs, which results in an adaptivity bound of $4$, while their analysis only needs two 
	independent realizations to achieve a bound of $2$.
We now provide the technical proof of Theorem~\ref{theorem:adaptive-gap-upper}.
We first formally define the decision tree representation.

\begin{definition}[Decision tree representation for adaptive policy]
	\label{definition:decision-tree-adaptive-policy}
	An adaptive policy $\pi$ can be seen as a decision tree $\T(\pi)$,
	where each node $s$ of $\T(\pi)$ corresponds to a partial realization $\psi_s$, 
	with the root being the empty partial realization, and node $s'$ is a child of $s$
	if $\psi_{s'} = \psi_{s} \cup \{\pi(\psi_s), \phi(\pi(\psi_s)) \}$ for some realization
	$\phi \sim \psi_s$. 
	Each node $s$ is associated with a probability $p_s$, which is the probability that the policy $\pi$ generates partial realization $\psi_s$, i.e. the probability that 
	the policy would walk on the tree from the root to node $s$.
	
\end{definition}


Next we define the non-adaptive randomized policy $\cW(\pi)$, which randomly selects
	a leaf of $\T(\pi)$.

\begin{definition}[Random-walk non-adaptive policy \cite{gupta2017adaptivity}]
	\label{definition:random-non-adaptive}
	For any adaptive policy $\pi$, let $\aL(\pi)$ denote the set of leaves of $\T(\pi)$. Then we construct a randomized non-adaptive policy $\cW(\pi)$ as follows: for any leaf $\ell \in \aL(\pi)$, $\cW(\pi)$ picks leaf $\ell$ with probability $p_\ell$ and selects $\dom(\psi_{\ell})$ as the seed set.
\end{definition}


Before proceeding further with our analysis, we introduce some notations for the 
myopic feedback model.
In the myopic feedback model, we notice that the state spaces for all nodes are 
mutually independent and disjoint. 
Thus we could decompose the realization space $\cR$ into independent subspace, $\cR = 
\times_{u\in V} O_u$, where $O_{u}$ is the set of all possible states for node $u$. For any full realization $\phi$ (resp. partial realization $\psi$), we would use $\phi_{S}$ (resp. $\psi_{S}$) to denote the feedback for the node set $S \subseteq V$. 
Note that $\phi_S$ and $\psi_S$ are partial realizations with domain $S$.
Similarly, we would also use $\cP_{S}$ to denote the probability 
space $\times_{u \in S}\cP_{u}$, where $\cP_{u}$ is the probability distribution over $O_u$
(i.e. each out-going edge $(u,v)$ of $u$ is live with independent probability $p_{uv}$). 
With a slight abuse of notation, we further use 
$\phi_{S}$ (resp. $\psi_{S}$) to denote the set of live edges leaving from $S$
under $\phi$ (resp. $\psi$). 
Then we could use notation $\phi^1_S \cup \phi^2_S$ to represent the union of live-edges
from $\phi^1$ and $\phi^2$ leaving from $S$, and similarly  $\psi^1 \cup \phi^2_S$ with
$\dom(\psi) = S$.

\paragraph{Construction for hybrid policy.} For any adaptive policy $\pi$, we define 
a fictitious hybrid policy $\bp$ that works on
three independent random realizations $\Phi^1$, $\Phi^2$ and $\Phi^3$ simultaneously, thinking about them
	as from three copies of the graphs $G_1$, $G_2$ and $G_3$.
Note that $\bp$ is not a real adaptive
policy --- it is only used for our analytical purpose to build connections between
the adaptive policy $\pi$ and the non-adaptive policy $\cW(\pi)$.
In terms of adaptive seed selection, $\bp$ acts exactly the same as $\pi$ on $G_1$, 
	responding to partial realizations $\psi^1$ obtained so far from the full realization
	$\Phi^1$ of $G_1$, and disregarding the realizations $\Phi^2$ and $\Phi^3$.
However, the difference is when we define adaptive influence spread for $\bp$, we aggregate
	the three partial realizations on the seed set together.
More precisely, 
%
%
%
%
	for any $t = 1,2,3$, we define the $t$-th aggregate influence utility function as
$f^{t}: 2^V \times \cR^{t} \rightarrow \R^{+}$
\begin{align}\label{eq: utility-fake}
f^{t}\left(S, \phi^1, \cdots, \phi^t\right) := f\left(S, (\cup_{i \in [t]}\phi^{i}_{S}, \phi^{1}_{V \backslash S})\right),
\end{align}
where $(\cup_{i \in [t]}\phi^{i}_{S}, \phi^{1}_{V \backslash S})$ means a new realization
$\phi'$ where on set $S$ its set of out-going live-edges is the same as union of $\phi^1, \cdots \phi^t$
and on set $V\setminus S$, its set of out-going live-edges is the same as $\phi^1$, and $f$ is the original influence utility function
	defined in Section~\ref{sec:model}.
The objective of the hybrid policy $\bp$ is then defined as the adaptive influence 
spread under policy $\bp$, i.e.,
\begin{align}
\label{eq: utility-hybrid-strategy}
\bfa(\bp) & := 
\E_{\Phi^1, \Phi^2, \Phi^3 \sim \cP}\left[f^{3}(V(\pi, \Phi^1),\Phi^1,\Phi^2, \Phi^3)\right] \nonumber \\
&=\E_{\Phi^1, \Phi^2, \Phi^3 \sim \cP}\left[f\left(V(\pi, \Phi^1), (\Phi^{1}_{V(\pi, \Phi^1)} \cup \Phi^{2}_{V(\pi, \Phi^1)}\cup \Phi^{3}_{V(\pi, \Phi^1)}, \Phi^{1}_{V\backslash V(\pi, \Phi^1)})\right)\right].
\end{align}

In other words, the adaptive influence spread of the hybrid policy $\bp$ is the influence spread of seed nodes 
	$V(\pi, \Phi^1)$ selected in graph $G_1$ by policy $\pi$, 
where the live-edge graph on the seed set part $V(\pi, \Phi^1)$ is the union of
live-edge graphs of $G_1$, $G_2$ and $G_3$, and the live-edge graph on the non-seed set part
is only that of $G_1$.
It can also be viewed as each seed node has three independent chances to activate its out-neighbors.
Since the hybrid policy $\bp$ acts the same as policy $\pi$ on influence graph $G_1$, we can easily conclude:
\begin{claim}
	\label{claim:hybrid-vs-adaptive}
	$\bfa(\bp) \geq \fa(\pi).$
\end{claim}

We also define $t$-th aggregate influence spread for a seed set $S$, $\sigma^t(S)$, as
		$ \sigma^{t}(S) = \E_{\Phi^1, \cdots, \Phi^t \sim \cP}\left[f^{t}(S,\Phi^1, \cdots, \Phi^t)\right]$.
Then, for the random-walk non-adaptive policy $\cW(\pi)$, we define
	$\sigma^t(\cW(\pi)) = \sum_{\ell \in \aL(\pi)} p_{\ell} \cdot \sigma^t(\dom(\psi_{\ell}))$,
	that is, the $t$-th aggregate influence spread of $\cW(\pi)$ is the average
	$t$-th aggregate influence spread of seed nodes selected by $\cW(\pi)$ according to 
	distribution of the leaves in the decision tree $\T(\pi)$.
{Similarly, we define the $t$-th aggregate adaptive influence spread for an adaptive policy $\pi$ as $ \sigma^{t}(\pi) = \E_{\Phi^1, \cdots, \Phi^t \sim \cP}\left[f^{t}(V(\pi, \Phi^1),\Phi^1, \cdots, \Phi^t)\right]$.} 
Note that $\bfa(\bp) = \sigma^3(\pi)$.

Now, we could give a definition for the conditional expected marginal gain for the 
	aggregate influence utility function $f^t$ over live-edge graph distributions.
\begin{definition}
	\label{definition:conditional-exptected-marginal-utility-hybrid}
	The expected non-adaptive marginal gain of $u$ given set $S$ 
	under $f^t$ is defined as:
	\begin{align}
	\label{eq:exptected-marginal-utility-hybrid}
	\Delta_{f^t}(u \mid S) = \E_{\Phi^1, \cdots, \Phi^t \sim \cP}\left[f^t\left(S \cup \{u\}, \Phi^1, \cdots, \Phi^{t}\right) - f^t\left(S, \Phi^1, \cdots, \Phi^{t}\right)\right].
	\end{align}
	The expected adaptive marginal gain of $u$ given partial realization 
		$\psi^1$  under $f^t$ is defined as: 
	\begin{align}
	\label{eq:conditional-exptected-marginal-utility-hybrid}
	\Delta_{f^t}(u \mid \psi^1) = \E_{\Phi^1, \cdots, \Phi^t \sim \cP}\left[f^t\left(\dom(\psi^1) \cup \{u\}, \Phi^1, \cdots, \Phi^{t}\right) - f^t\left(\dom(\psi^1), \Phi^1, \cdots, \Phi^{t}\right) \mid \Phi^1 \sim \psi^1\right].
	\end{align}
\end{definition}

The following lemma connects $\sigma^t(\pi)$ (and thus $\bfa(\bp)$) with 
	adaptive marginal gain $\Delta_{f^t}(u \mid \psi)$, and
	connects $\sigma^t(\cW(\pi))$ with non-adaptive marginal gain $\Delta_{f^t}(u \mid S)$.
{Let 
		$\cP_{i}^{\pi}$ denote the probability distribution over 
		nodes at  depth $i$ of the decision $\T(\pi)$.}
The proof is by applying telescoping series to convert influence spread into the sum of marginal gains.

\begin{restatable}{lemma}{lemmarginal} 
\label{lem:marginal}
For any adaptive policy $\pi$, and $t \geq 1$, we have
\begin{align*}
\sigma^{t}(\pi) = \sum_{i = 0}^{k-1}\E_{s \sim \cP_{i}^{\pi}}\left[\Delta_{f^{t}}\left(\pi(\psi_s) \mid \psi_s\right)\right], \mbox{ and }
\sigma^t(\cW(\pi)) = \sum_{i = 0}^{k-1}\E_{s \sim \cP_{i}^{\pi}}\left[\Delta_{f^{t}}\left(\pi(\psi_s) \mid \dom(\psi_s)\right)\right].
\end{align*}
\end{restatable}
The next lemma bounds two intermediate adaptive marginal gains to be used
	for Lemma~\ref{lemma:hybrid-vs-non-adaptive}.
The proof crucially depend on (a) the independence of realizations $\Phi^1,\Phi^2,\Phi^3$,
	(b) the independence of feedback of different selected seed nodes, and
	(c) the submodularity of the influence utility function
	on live-edge graphs.


\begin{restatable}{lemma}{lemAdaptiveGapMyopicOne}
	\label{lemma:adaptive-gap-myopic1}
	Let $S = \dom(\psi^1)$ and $\Sps = S \cup \{u\}$ for 
	any partial realization $\psi^1$ and any $u\not\in \dom(\psi^1)$. Then 
		we have
	\begin{align}
	&\E_{\Phi^1, \Phi^2, \Phi^3 \sim \cP}\left[f\left(\Sps, (\Phi^{1}_{S} \cup \Phi^{2}_{S} \cup \Phi^{3}_{S}, \Phi^{1}_{u} \cup \Phi^{2}_{u}, \Phi^{1}_{V \backslash \Sps}) \right) \right.\nonumber \\ 
	& \qquad - \left. f\left(S, (\Phi^{1}_{S} \cup \Phi^{2}_{S} \cup \Phi^{3}_{S}, \Phi^{1}_{V \backslash S}) \right) \mid \Phi^1 \sim \psi^1\right] \leq \Delta_{f^2}(u \mid S).\label{eq:adaptive-gap-upper5}
	\end{align}
	\begin{align}
&\E_{\Phi^1, \Phi^2, \Phi^3 \sim \cP}\left[f\left(\Sps,(\Phi^{1}_{S} \cup \Phi^{2}_{S} \cup \Phi^{3}_{S}, \Phi^{1}_{u} \cup \Phi^{2}_{u} \cup \Phi^{3}_{u} , \Phi^{1}_{V \backslash \Sps}) \right)\right.\nonumber\\
&\qquad - \left.f\left(\Sps, (\Phi^{1}_{S} \cup \Phi^{2}_{S} \cup \Phi^{3}_{S}, \Phi^{1}_{u} \cup \Phi^{2}_{u}, \Phi^{1}_{V \backslash \Sps}) \right) \mid \Phi^1 \sim \psi^1\right]	
\leq \Delta_{f^2}(u \mid S).\label{eq:adaptive-gap-upper6}
\end{align}
\end{restatable}

Combining the two inequalities above, we obtain the following key lemma essential, 
	which bounds the adaptive marginal gain
	$\Delta_{f^{3}}(u \mid \psi^1) $ with the 
	non-adaptive marginal gain $\Delta_{f^{2}}(u \mid \dom(\psi^1))$.
\begin{restatable}{lemma}{lemHybridvsNonAdaptive}
	\label{lemma:hybrid-vs-non-adaptive}
	For any partial realization $\psi^1$ and node $u \notin \dom(\psi^1)$, we have
	\begin{align}
	\label{eq:adaptive-gap-upper13}
	\Delta_{f^{3}}(u \mid \psi^1) \leq 2\Delta_{f^{2}}(u \mid \dom(\psi^1)).
	\end{align}
\end{restatable}
The next lemma gives an upper bound on the $t$-th aggregate (non-adaptive) influence spread 
	$\sigma^{t}(S)$  using the original influence spread $\sigma(S)$.
The idea of the proof is that each seed node in $S$ has $t$ independent chances to active its 
	out-neighbors, but afterwards the diffusion is among nodes not in $S$ as in the original
	diffusion.
\begin{restatable}{lemma}{leminfluenceSpreadHighOrder}
	\label{lemma:influence-spread-high-order}
	For any $t \geq 1$ and any subset $S\subseteq V$, 
	$\sigma^{t}(S)  \leq t\cdot \sigma(S)$.
\end{restatable}

\begin{proof}[Proof of Theorem \ref{theorem:adaptive-gap-upper}]
It is enough to show that for every adaptive policy $\pi$, 
	$\fa(\pi) \leq 4\fa(\cW(\pi))$.
%
This is done by the following derivation sequence:
$\fa(\pi) \leq \bfa(\bp) = \sigma^3(\pi)
	= \sum_{i = 0}^{k-1}\E_{s \in \cP_{i}^{\pi}}\left[\Delta_{f^{3}}\left(\pi(\psi_s) \mid \psi_s\right)\right]
	\leq\sum_{i = 0}^{k-1}\E_{s \in \cP_{i}^{\pi}}\left[2\Delta_{f^2}\left(\pi(\psi_s) \mid \dom(\psi_s)\right)\right]
	=2\fa^2(\cW(\pi)) \leq 4\fa(\cW(\pi))$,
	where the first inequality is by Claim \ref{claim:hybrid-vs-adaptive}, the second
	and the third equalities are by Lemma~\ref{lem:marginal}, 
	the second inequality is by Lemma \ref{lemma:hybrid-vs-non-adaptive} and the
	last inequality is by Lemma~\ref{lemma:influence-spread-high-order}.
%
\end{proof}

\subsection{Lower bound}
\label{sec:adaptive-gap-lower-bound}
Next, we proceed to give a lower bound on the adaptivity gap for the influence maximization problem in the myopic feedback model. Our result is stated as follow:
\begin{restatable}{theorem}{thmAdaptiveGapLower}
	\label{theorem:adaptive-gap-lower}
	Under the IC model with myopic feedback, the adaptivity gap for the influence
	maximization problem is at least $e/(e - 1)$.
\end{restatable}
\begin{proof}[Proof Sketch]
We construct a bipartite graph $G = (L, R, E, p)$ with $|L| = \binom{m^3}{m^2}$ and $|R|=m^3$.
For each subset $X\subset R$ with $|X|=m^2$, there is exactly one node $u\in L$ that connects to all nodes
	in $X$.
We show that for any $\epsilon > 0$, there is a large enough $m$ such that 
	in the above graph with parameter $m$ the adaptivity gap is at least $e/(e - 1)-\epsilon$.
\end{proof}

\section{Adaptive and Non-Adaptive Greedy Algorithms}
\label{sec:constant-competitive-algo}

In this section, we consider two prevalent algorithms --- the \emph{greedy} algorithm and the \emph{adaptive greedy} algorithm for the influence maximization problem. 
To the best of our knowledge, we provide the first approximation ratio of these algorithms 
	with respect to the adaptive optimal solution in the IC model with myopic feedback.
%
%
We formally describe the algorithms in figure \ref{fig:description-greedy-adaptive-greedy}. 

\begin{figure}[!htbp]
	\begin{tabular}{|*2{>{\centering\arraybackslash}p{0.45\linewidth}|}}
		\hline
		\begin{algorithmic}
			\label{algo:greedy}
			\STATE {\bf Greedy Algorithm:}
			\STATE $S = \emptyset$
			\WHILE{$|S| < k$}
			\STATE$u = \argmax_{u \in V\backslash S} \Delta_{f}(u | S)$\\
			\STATE $S = S \cup \{u\}$\\
			\ENDWHILE
			\RETURN $S$
		\end{algorithmic}
		&
		\begin{algorithmic}
			\label{algo:adaptive-greedy}
			\STATE {\bf Adaptive Greedy Algorithm:}
			\STATE $S = \emptyset, \Psi = \emptyset$
			\WHILE{$|S| < k$}
			\STATE$u = \argmax_{u \in V\backslash S} \Delta_{f}(u | \Psi)$
			\STATE Select $u$ as seed and observe $\Phi(u)$.
			\STATE $S = S \cup \{u\}$, $\Psi = \Psi \cup \{(u, \Phi(u))\}$ 
			\ENDWHILE
		\end{algorithmic}\\
		\hline
	\end{tabular}
	\caption{Description for \emph{greedy} and \emph{adaptive greedy}.}
	\label{fig:description-greedy-adaptive-greedy}
\end{figure}

Our main result is summarized below.
\begin{theorem}
	\label{theorem:greedy&adaptive-greedy-constant}
	Both greedy and adaptive greedy are $\frac{1}{4}(1 - \frac{1}{e})$ approximate to the optimal adaptive policy under the IC model with myopic feedback.
\end{theorem}
\begin{proof}[Proof Sketch]
The proof for the non-adaptive greedy algorithm is straightforward since the greedy algorithm 
	provides a $(1 - \frac{1}{e})$ approximation to the non-adaptive optimal solution, which
	by Theorem \ref{theorem:adaptive-gap-upper} is at least $\frac{1}{4}$ of the adaptive
	optimal solution.
For the adaptive greedy algorithm, we need to separately prove that it also provides a
	$(1 - \frac{1}{e})$ approximation to the non-adaptive optimal solution, and then
	the result is immediate similar to the non-adaptive greedy algorithm.
\end{proof}

Theorem~\ref{theorem:greedy&adaptive-greedy-constant} shows that greedy and adaptive greedy
	can achieve at least an approximation ratio of $\frac{1}{4}(1 - \frac{1}{e})$ 
	with respect to the adaptive optimal solution.
We further show that their approximation ratio is at most $\frac{e^2 + 1}{(e + 1)^2}\approx 0.606$,
	which is strictly less than $1-1/e \approx 0.632$.
To do so, we first present an example for non-adaptive greedy with approximation ratio
at most $\frac{e^2 + 1}{(e + 1)^2}$.
Next, we show that myopic feedback does not help much to adaptive greedy, in that
	the approximation ratio for the non-adaptive greedy algorithm is no worse than adaptive greedy, when considering over all graphs. 
Combining with the first observation, we also achieve the result for the adaptive greedy algorithm.
\begin{theorem}
	\label{theorem:bad-example-greedy}
	The approximation ratio for greedy and adaptive greedy is no better than $\frac{e^2 + 1}{(e + 1)^2}\approx 0.606$, which is strictly less than $1-1/e \approx 0.632$.
	Moreover, the approximation ratio of adaptive greedy is at most that of the non-adaptive greedy,
	when considering all influence graphs.
%
\end{theorem}

\section{Conclusion and Future Work}

In this paper, we systematically study the adaptive influence maximization problem with 
	myopic feedback under the independent cascade model, and provide constant upper and lower bounds
	on the adaptivity gap and the approximation ratios of the non-adaptive greedy and adaptive
	greedy algorithms.
There are a number of future directions to continue this line of research.
First, there is still a gap between the upper and lower bound results in this paper, and
	thus how to close this gap is the next challenge.
Second, our result suggests that adaptive greedy may not bring much benefit under the myopic feedback
	model, so are there other adaptive algorithms that could do much better?
Third, for the IC model with full-adoption feedback, because the feedback on different seed nodes
	may be correlated, existing adaptivity gap results in \cite{asadpour2015maximizing,bradac2019near}
	cannot be applied, and thus its adaptivity gap is still open 
	even though it is adaptive submodular.
One may also explore beyond the IC model, and study adaptive solutions for other models such as
	the linear threshold model, general threshold model etc.\cite{Kempe2003maximizing}.
Finally, scalable algorithms for adaptive influence maximization is also worth to investigate.

\section*{Acknowledgment}
We would like to thank Dimitris Kalimeris for some early discussions on the subject.
Wei Chen is partially supported by the National Natural Science Foundation of China (Grant No.
61433014).

\bibliographystyle{plainnat}
\bibliography{myopic_adaptive}

\newpage
\appendix
\section*{Appendix}
We include the missing proofs in this appendix. 
For convenience, we restate the lemmas and theorems that we prove here.

\section{Missing Proofs of Section~\ref{sec:adaptivityUpperBound}, Adaptivity Upper Bound}

{\lemmarginal*}
\begin{proof}
	%
	We first prove the equality on $\sigma^t(\pi)$.
	Let $V(\pi, \Phi)_{:i}$ (resp. $V(\pi, \Phi)_{i}$) denote the first $i$ nodes 
	(resp. the $i^{th}$ node) selected by policy $\pi$ under realization $\Phi$.
	
	Then we have
	\begin{align}
	&\sum_{i = 0}^{k-1}\E_{s \sim \cP_{i}^{\pi}}\left[\Delta_{f^{t}}\left(\pi(\psi_s) \mid \psi_s\right)\right] \notag\\
	=&\sum_{i = 0}^{k-1}\E_{s \sim \cP_{i}^{\pi}}\left[\E_{\Phi^1, \cdots, \Phi^t \sim \cP}\left[f^t\left(\dom(\psi_s) \cup \pi(\psi_s), \Phi^1, \cdots, \Phi^{t}\right) - f^t\left(\dom(\psi_s), \Phi^1, \cdots, \Phi^{t}\right) \mid \Phi^1 \sim \psi_s\right]\right] \notag\\
	=&\sum_{i = 0}^{k-1}\E_{\Phi^2, \cdots, \Phi^t \sim \cP}\left[\E_{s \sim \cP_{i}^{\pi}}\left[\E_{\Phi^1 \sim \cP}\left[f^t\left(\dom(\psi_s) \cup \pi(\psi_s), \Phi^1, \cdots, \Phi^{t}\right) - f^t\left(\dom(\psi_s), \Phi^1, \cdots, \Phi^{t}\right) \mid \Phi^1 \sim \psi_s\right]\right]\right] \notag\\
	=&\sum_{i = 0}^{k-1}\E_{\Phi^2, \cdots, \Phi^t \sim \cP}\left[\E_{\Phi^1 \sim \cP}\left[\left(f^t\left(V(\pi, \Phi^1)_{:i} \cup V(\pi, \Phi^1)_{i + 1}, \Phi^1, \cdots, \Phi^{t}\right) - f^t\left(V(\pi, \Phi^1)_{:i}, \Phi^1, \cdots, \Phi^{t}\right)\right) \right] \right]\notag\\
	=&\E_{\Phi^2, \cdots, \Phi^t \sim \cP}\left[\E_{\Phi^1 \sim \cP}\left[\sum_{i = 0}^{k-1}\left(f^t\left(V(\pi, \Phi^1)_{:i} \cup V(\pi, \Phi^1)_{i + 1}, \Phi^1, \cdots, \Phi^{t}\right) - f^t\left(V(\pi, \Phi^1)_{:i}, \Phi^1, \cdots, \Phi^{t}\right)\right) \right] \right]\notag\\
	=& \E_{\Phi^2, \cdots, \Phi^t \sim \cP}\left[\E_{\Phi^1\sim \cP}\left[f^{t}(V(\pi, \Phi^1),\Phi^1, \cdots, \Phi^t)\right]\right]\notag\\
	=&\sigma^{t}(\pi) \nonumber 
	\end{align}
	The third equality above is by the law of total expectation, and notice that
	for any tree node $s$ in $\T(\pi)$ and 
	any random realization $\Phi \sim \psi_s$, we have $V(\pi, \Phi)_{:i} = \dom(\psi_s)$ and $V(\pi, \Phi)_{i + 1} = \pi(\psi_s)$.
	
	
	Next, we prove the equality on $\sigma^t(\cW(\pi))$.
	\begin{align}
	&\sum_{i = 0}^{k-1}\E_{s \sim \cP_{i}^{\pi}}\left[\Delta_{f^{t}}\left(\pi(\psi_s) \mid \dom(\psi_s)\right)\right] \notag\\
	&=\sum_{i = 0}^{k-1}\E_{s \sim \cP_{i}^{\pi}}\left[\E_{\Phi^1, \cdots, \Phi^t \sim \cP}\left[f^t\left(\dom(\psi_s) \cup \pi(\psi_s), \Phi^1, \cdots, \Phi^{t}\right) - f^t\left(\dom(\psi_s), \Phi^1, \cdots, \Phi^{t}\right)\right]\right] \notag\\
	&=\sum_{i = 0}^{k-1}\E_{\Phi^1, \cdots, \Phi^t \sim \cP}\left[\E_{s \sim \cP_{i}^{\pi}}\left[f^t\left(\dom(\psi_s) \cup \pi(\psi_s), \Phi^1, \cdots, \Phi^{t}\right) - f^t\left(\dom(\psi_s), \Phi^1, \cdots, \Phi^{t}\right) \right]\right]\notag\\
	&=\sum_{i = 0}^{k-1}\E_{\Phi^1, \cdots, \Phi^t \sim \cP}\left[\E_{\Phi\sim\cP}\left[f^t\left(V(\pi, \Phi)_{:i} \cup V(\pi, \Phi)_{i + 1}, \Phi^1, \cdots, \Phi^{t}\right) - f^t\left(V(\pi, \Phi)_{:i}, \Phi^1, \cdots, \Phi^{t}\right) \right]\right]\notag\\
	&=\E_{\Phi\sim\cP}\left[\E_{\Phi^1, \cdots, \Phi^t \sim \cP}\left[\sum_{i = 0}^{k-1}\left(f^t\left(V(\pi, \Phi)_{:i} \cup V(\pi, \Phi)_{i + 1}, \Phi^1, \cdots, \Phi^{t}\right) - f^t\left(V(\pi, \Phi)_{:i}, \Phi^1, \cdots, \Phi^{t}\right)\right) \right]\right]\notag\\
	&= \E_{\Phi \sim \cP}\left[\E_{\Phi^1, \cdots, \Phi^t \sim \cP}\left[f^{t}(V(\pi, \Phi),\Phi^1, \cdots, \Phi^t)\right]\right]\notag\\
	&= \E_{\Phi \sim \cP}\left[\sigma^t(V(\pi, \Phi)) \right]
	\notag\\
	&=\sigma^{t}(\cW(\pi)).\notag
	\end{align}
	The third equality above is because the distribution of $\dom(\psi_s)$ and $\pi(\psi_s)$
	with $s \sim \cP_{i}^{\pi}$ is exactly the same as the distribution of 
	$V(\pi, \Phi)_{:i}$ and $V(\pi, \Phi)_{i+1} $ with $\Phi \sim \cP$.
	Note that this $\Phi$ is independent of $\Phi^1, \cdots, \Phi^t$.
	The last equality is because the distribution of $V(\pi, \Phi)$ with $\Phi \sim \cP$ is exactly the
	distribution of the seed sets taken from the leaves of $\T(\pi)$, which exactly
	corresponds to the random-walk non-adaptive policy $\cW(\pi)$.
	%
\end{proof}

{\lemAdaptiveGapMyopicOne*}
\begin{proof}
	We first prove Inequality \eqref{eq:adaptive-gap-upper5}.
	To do so, we first expand the RHS of Eq.~\eqref{eq:adaptive-gap-upper5},
	\begin{align}
	&\Delta_{f^2}(u \mid S) = \E_{\Phi^2, \Phi^3 \sim \cP}\left[f^2\left(\Sps, \Phi^2, \Phi^{3}\right) - f^2\left(S, \Phi^2, \Phi^{3}\right)\right]\notag\\
	& = \E_{\Phi^2, \Phi^3 \sim \cP}\left[f\left(\Sps, (\Phi^2_{S}\cup \Phi^{3}_{S},\Phi^2_{u}\cup \Phi^{3}_{u}, \Phi^{2}_{V \backslash \Sps})  \right) - f\left(S, (\Phi^2_{S}\cup \Phi^{3}_{S},\Phi^2_{u}, \Phi^{2}_{V \backslash \Sps})\right)\right]\notag\\
	& = \E_{\Phi^{2}_{S}, \Phi^{3}_{S} \sim \cP_{S}}\left[\E_{\Phi^{2}_{u}, \Phi^{3}_{u} \in \cP_{u} }\left[\E_{\Phi^{2}_{V\backslash \Sps} \sim \cP_{V \backslash \Sps}}\left[f\left(\Sps, (\Phi^2_{S}\cup \Phi^{3}_{S},\Phi^2_{u}\cup \Phi^{3}_{u}, \Phi^{2}_{V \backslash \Sps})  \right) - 
	\right. \right. \right.	\nonumber \\
	& \qquad \qquad  \left. \left. \left. f\left(S, (\Phi^2_{S}\cup \Phi^{3}_{S},\Phi^2_{u}, \Phi^{2}_{V \backslash \Sps})\right)\right]\right]\right].\label{eq:adaptive-gap-upper1}
	\end{align}
	The third equality above holds because $\Phi^{2}_{S}, \Phi^{3}_{S}, \Phi^{2}_{u}, \Phi^{3}_{u}, \Phi^{2}_{V \backslash \Sps}, \Phi^{3}_{V \backslash \Sps}$ are mutually independent, and
	$\Phi^{3}_{V \backslash \Sps}$ does not appear inside the expectation term.
	Next, we expand the LHS of Eq.~\eqref{eq:adaptive-gap-upper5},
	\begin{align}
	&\mbox{LHS of Eq.~\eqref{eq:adaptive-gap-upper5}} \nonumber \\
	&=\E_{\Phi^{1}_S, \Phi^2_{S}, \Phi^{3}_S \sim \cP_S}\left[\E_{\Phi^{1}_{u}, \Phi^{2}_{u} \in \cP_{u} }\left[\E_{\Phi^{1}_{V\backslash \Sps} \sim \cP_{V \backslash \Sps}}\left[f\left(\Sps,(\Phi^{1}_{S} \cup \Phi^{2}_{S} \cup \Phi^{3}_{S}, \Phi^{1}_{u} \cup \Phi^{2}_{u}, \Phi^{1}_{V \backslash \Sps}) \right) \right.\right.\right.\nonumber\\
	&\qquad - \left.\left.\left.f\left(S, (\Phi^{1}_{S} \cup \Phi^{2}_{S} \cup \Phi^{3}_{S}, \Phi^{1}_{u}, \Phi^{1}_{V \backslash \Sps}) \right) \mid \Phi^1 \sim \psi^1\right]\right]\right]\nonumber\\
	&=\E_{\Phi^2_{S}, \Phi^{3}_S \sim \cP_S}\left[\E_{\Phi^{1}_{u}, \Phi^{2}_{u} \in \cP_{u} }\left[\E_{\Phi^{1}_{V\backslash \Sps} \sim \cP_{V \backslash \Sps}}\left[f\left(\Sps,(\psi^1 \cup \Phi^{2}_{S} \cup \Phi^{3}_{S}, \Phi^{1}_{u} \cup \Phi^{2}_{u}, \Phi^{1}_{V \backslash \Sps}) \right) \right.\right.\right.\nonumber\\
	&\qquad - \left.\left.\left.f\left(S, (\psi^1 \cup \Phi^{2}_{S} \cup \Phi^{3}_{S}, \Phi^{1}_{u}, \Phi^{1}_{V \backslash \Sps}) \right)\right]\right]\right].\nonumber\\
	&=\E_{\Phi^2_{S}, \Phi^{3}_S \sim \cP_S}\left[\E_{\Phi^{1}_{u}, \Phi^{3}_{u} \in \cP_{u} }\left[\E_{\Phi^{2}_{V\backslash \Sps} \sim \cP_{V \backslash \Sps}}\left[f\left(\Sps,(\psi^1 \cup \Phi^{2}_{S} \cup \Phi^{3}_{S}, \Phi^{1}_{u} \cup \Phi^{3}_{u}, \Phi^{2}_{V \backslash \Sps}) \right) \right.\right.\right.\nonumber\\
	&\qquad - \left.\left.\left.f\left(S, (\psi^1 \cup \Phi^{2}_{S} \cup \Phi^{3}_{S}, \Phi^{1}_{u}, \Phi^{2}_{V \backslash \Sps}) \right)\right]\right]\right].\notag\\
	&=\E_{\Phi^2_{S}, \Phi^{3}_S \sim \cP_S}\left[\E_{\Phi^{2}_{u}, \Phi^{3}_{u} \in \cP_{u} }\left[\E_{\Phi^{2}_{V\backslash \Sps} \sim \cP_{V \backslash \Sps}}\left[f\left(\Sps,(\psi^1 \cup \Phi^{2}_{S} \cup \Phi^{3}_{S}, \Phi^{2}_{u} \cup \Phi^{3}_{u}, \Phi^{2}_{V \backslash \Sps}) \right) \right.\right.\right.\nonumber\\
	&\qquad - \left.\left.\left.f\left(S, (\psi^1 \cup \Phi^{2}_{S} \cup \Phi^{3}_{S}, \Phi^{2}_{u}, \Phi^{2}_{V \backslash \Sps}) \right)\right]\right]\right].\label{eq:adaptive-gap-upper9}
	\end{align}
	The first equality above holds because all these random variables are independent.
	The second equality above holds because $\Phi^{1}_{S} = \psi^1$ implied by $\Phi^1 \sim \psi^1$. 
	In the third equality, we replace $\Phi^{1}_{V \backslash \Sps}$ with $\Phi^{2}_{V \backslash \Sps}$ and replace $\Phi^{2}_{u}$ with $\Phi^{3}_{u}$, 
	because they follow the same probability distributions and are independent to the other
	distributions. In the last equality, we replace $\Phi^{1}_{u}$ with $\Phi^{2}_{u}$.

	%
	Comparing Eq.~\eqref{eq:adaptive-gap-upper1} and Eq.~\eqref{eq:adaptive-gap-upper9}, we know that it suffices to prove that for any fixed partial realizations $\phi^{2}_{S}, \phi^{3}_{S}, \phi^{2}_{u}, \phi^{3}_{u}, \phi^{2}_{V \backslash S}$,
	\begin{align}
	&f\left(\Sps,(\psi^1 \cup \phi^{2}_{S} \cup \phi^{3}_{S}, \phi^{2}_{u} \cup \phi^{3}_{u}, \phi^{2}_{V \backslash \Sps})\right) - f\left(S, (\psi^1 \cup \phi^{2}_{S} \cup \phi^{3}_{S}, \phi^{2}_{u}, \phi^{2}_{V \backslash \Sps}) \right)\nonumber\\
	& \leq f\left(\Sps,(\phi^{2}_{S} \cup \phi^{3}_{S}, \phi^{2}_{u} \cup \phi^{3}_{u}, \phi^{2}_{V \backslash \Sps})\right) - f\left(S, (\phi^{2}_{S} \cup \phi^{3}_{S}, \phi^{2}_{u}, \phi^{2}_{V \backslash \Sps}) \right).\label{eq:adaptive-gap-upper8}
	\end{align}
	Consider any node $v \in \Gamma(\Sps,(\psi^1 \cup \phi^{2}_{S} \cup \phi^{3}_{S}, \phi^{2}_{u} \cup \phi^{3}_{u}, \phi^{2}_{V \backslash \Sps})) \backslash \Gamma(S, (\psi^1 \cup \phi^{2}_{S} \cup \phi^{3}_{S}, \phi^{2}_{u}, \phi^{2}_{V \backslash \Sps}))$, we have the following observations: 
	(1) under the realization $(\psi^1 \cup \phi^{2}_{S} \cup \phi^{3}_{S}, \phi^{2}_{u}, \phi^{2}_{V \backslash \Sps})$ (or equivalently its live-edge graph), node $v$ cannot be reached from nodes in $S$;  and
	(2) under the realization $(\psi^1 \cup \phi^{2}_{S} \cup \phi^{3}_{S}, \phi^{2}_{u}\cup\phi^{3}_{u}, \phi^{2}_{V \backslash \Sps})$ (or equivalently its live-edge graph), node $v$ can be reached via a path $P$ originated from node $u$, and $P$ does not contain any node in $S$. 
	
	Now, we are going to prove that $v \in \Gamma(\Sps,(\phi^{2}_{S} \cup \phi^{3}_{S}, \phi^{2}_{u} \cup \phi^{3}_{u}, \phi^{2}_{V \backslash \Sps})) \backslash \Gamma(S, (\phi^{2}_{S} \cup \phi^{3}_{S}, \phi^{2}_{u}, \phi^{2}_{V \backslash \Sps}) )$. Since the path $P$ does not contain any node in $S$, we know that path $P$ also exists under the realization $(\phi^{2}_{S} \cup \phi^{3}_{S}, \phi^{2}_{u} \cup \phi^{3}_{u}, \phi^{2}_{V \backslash \Sps})$, i.e., node $v$ can be reached from node $u$ under realization $(\phi^{2}_{S} \cup \phi^{3}_{S}, \phi^{2}_{u} \cup \phi^{3}_{u}, \phi^{2}_{V \backslash \Sps})$. Moreover, we know that the realization $((\phi^{2}_{S} \cup \phi^{3}_{S}, \phi^{2}_{u}, \phi^{2}_{V \backslash \Sps})$ has less live edges than the realization $(\psi^1 \cup \phi^{2}_{S} \cup \phi^{3}_{S}, \phi^{2}_{u}, \phi^{2}_{V \backslash \Sps})$, so node $v$ can not be reached from set $S$ under the realization $(\phi^{2}_{S} \cup \phi^{3}_{S}, \phi^{2}_{u}, \phi^{2}_{V \backslash \Sps})$ .
	As a result, we have proved
	\begin{align}
	&\Gamma\left(\Sps,(\psi^1 \cup \phi^{2}_{S} \cup \phi^{3}_{S}, \phi^{2}_{u} \cup \phi^{3}_{u}, \phi^{2}_{V \backslash \Sps})\right) \backslash \Gamma\left(S, (\psi^1 \cup \phi^{2}_{S} \cup \phi^{3}_{S}, \phi^{2}_{u}, \phi^{2}_{V \backslash \Sps})\right)\notag\\
	& \subseteq \Gamma\left(\Sps,(\phi^{2}_{S} \cup \phi^{3}_{S}, \phi^{2}_{u} \cup \phi^{3}_{u}, \phi^{2}_{V \backslash \Sps})\right) \backslash \Gamma\left(S, (\phi^{2}_{S} \cup \phi^{3}_{S}, \phi^{2}_{u}, \phi^{2}_{V \backslash \Sps}) \right).	\label{eq:adaptive-gap-upper10}
	\end{align}
	This proves Eq.~\eqref{eq:adaptive-gap-upper8} and thus concludes the proof
		of Inequality~\eqref{eq:adaptive-gap-upper5}.
	Note that the above proof on Eq.~\eqref{eq:adaptive-gap-upper8} resembles the proof of
	submodularity of influence utility function $f$ on a live-edge graph, but 
	Eq.~\eqref{eq:adaptive-gap-upper8} is a bit more complicated because it is on different
	live-edge graphs.

	Next we prove the Inequality \eqref{eq:adaptive-gap-upper6}.
	Again, we first expand the RHS of Eq.~\eqref{eq:adaptive-gap-upper6}.
	\begin{align}
	&\Delta_{f^2}(u \mid S) 
	= \E_{\Phi^{2}_{S}, \Phi^{3}_{S} \sim \cP_{S}}\left[\E_{\Phi^{2}_{u}, \Phi^{3}_{u} \in \cP_{u} }\left[\E_{\Phi^{2}_{V\backslash \Sps} \sim \cP_{V \backslash \Sps}}\left[f\left(\Sps, (\Phi^2_{S}\cup \Phi^{3}_{S},\Phi^2_{u}\cup \Phi^{3}_{u}, \Phi^{2}_{V\setminus \Sps})  \right) \right.\right.\right.
	\nonumber \\
	&\qquad \left.\left.\left. - f\left(S, (\Phi^2_{S}\cup \Phi^{3}_{S},\Phi^2_{u}, \Phi^{2}_{V \backslash \Sps})\right)\right]\right]\right]\notag\\
	&\geq\E_{\Phi^{2}_{S}, \Phi^{3}_{S} \sim \cP_{S}}\left[\E_{\Phi^{2}_{u}, \Phi^{3}_{u} \in \cP_{u} }\left[\E_{\Phi^{2}_{V\backslash \Sps} \sim \cP_{V \backslash \Sps}}\left[f\left(\Sps, (\Phi^2_{S}\cup \Phi^{3}_{S},\Phi^2_{u}\cup \Phi^{3}_{u}, \Phi^{2}_{V\setminus \Sps})  \right) \right.\right.\right.
	\nonumber \\
	& \qquad \left.\left.\left. - f\left(\Sps, (\Phi^2_{S}\cup \Phi^{3}_{S},\Phi^2_{u}, \Phi^{2}_{V \backslash \Sps})\right)\right]\right]\right].
	\label{eq:adaptive-gap-upperRHS2}
	\end{align}
	The inequality above is by the monotonicity of $f(S,\phi)$ on $S$. 
	Next, we expand the LHS of Eq.~\eqref{eq:adaptive-gap-upper6}.
	\begin{align}
	&\mbox{LHS of Eq.~\eqref{eq:adaptive-gap-upper6}} \nonumber\\
	&=\E_{\Phi^{1}_S, \Phi^2_{S}, \Phi^{3}_S \sim \cP_S}\left[\E_{\Phi^{1}_{u}, \Phi^{2}_{u}, \Phi^{3}_{u} \in \cP_{u} }\left[\E_{\Phi^{1}_{V\backslash \Sps} \sim \cP_{V \backslash \Sps}}\left[f\left(\Sps,(\Phi^{1}_{S} \cup \Phi^{2}_{S} \cup \Phi^{3}_{S}, \Phi^{1}_{u} \cup \Phi^{2}_{u} \cup \Phi^{3}_{u}, \Phi^{1}_{V \backslash \Sps}) \right) \right.\right.\right.\nonumber\\
	&\qquad - \left.\left.\left.f\left(\Sps, (\Phi^{1}_{S} \cup \Phi^{2}_{S} \cup \Phi^{3}_{S}, \Phi^{1}_{u} \cup \Phi^{2}_{u}, \Phi^{1}_{V \backslash \Sps}) \right) \mid \Phi^1 \sim \psi^1\right]\right]\right]\nonumber\\
	&=\E_{\Phi^2_{S}, \Phi^{3}_S \sim \cP_S}\left[\E_{\Phi^{1}_{u}, \Phi^{2}_{u}, \Phi^{3}_{u} \in \cP_{u} }\left[\E_{\Phi^{1}_{V\backslash \Sps} \sim \cP_{V \backslash \Sps}}\left[f\left(\Sps,(\psi^1 \cup \Phi^{2}_{S} \cup \Phi^{3}_{S}, \Phi^{1}_{u} \cup \Phi^{2}_{u} \cup \Phi^{3}_{u}, \Phi^{1}_{V \backslash \Sps}) \right) \right.\right.\right.\nonumber\\
	&\qquad - \left.\left.\left.f\left(\Sps, (\psi^1 \cup \Phi^{2}_{S} \cup \Phi^{3}_{S}, \Phi^{1}_{u} \cup \Phi^{2}_{u}, \Phi^{1}_{V \backslash \Sps}) \right)\right]\right]\right].\nonumber\\
	&=\E_{\Phi^2_{S}, \Phi^{3}_S \sim \cP_S}\left[\E_{\Phi^{1}_{u}, \Phi^{2}_{u}, \Phi^{3}_{u} \in \cP_{u} }\left[\E_{\Phi^{2}_{V\backslash \Sps} \sim \cP_{V \backslash \Sps}}\left[f\left(\Sps,(\psi^1 \cup \Phi^{2}_{S} \cup \Phi^{3}_{S}, \Phi^{1}_{u} \cup \Phi^{2}_{u} \cup \Phi^{3}_{u}, \Phi^{2}_{V \backslash \Sps}) \right) \right.\right.\right.\nonumber\\
	& \qquad - \left.\left.\left.f\left(\Sps, (\psi^1 \cup \Phi^{2}_{S} \cup \Phi^{3}_{S}, \Phi^{1}_{u} \cup \Phi^{2}_{u}, \Phi^{2}_{V \backslash \Sps}) \right)\right]\right]\right].\label{eq:adaptive-gap-upper11}
	\end{align}

	The last  equality holds by replacing $\Phi^{1}_{V\backslash \Sps}$ with $\Phi^{2}_{V\backslash \Sps}$, because both have the same distributions and are independent from the other distributions. 
	%
	%
	%
	Similar to the previous lemma, comparing Eq.~\eqref{eq:adaptive-gap-upperRHS2} and Eq.~\eqref{eq:adaptive-gap-upper11}, it suffices to prove that for fixed partial realizations $\phi^2_{S}, \phi^{3}_S, \phi^{1}_{u}, \phi^{2}_{u}, \phi^{3}_{u}$ and $\phi^{2}_{V \backslash \Sps}$, 
	\begin{align}
	&f\left(\Sps,(\psi^1 \cup \phi^{2}_{S} \cup \phi^{3}_{S}, \phi^{1}_{u} \cup \phi^{2}_{u} \cup \phi^{3}_{u}, \phi^{2}_{V \backslash \Sps}) \right) - f\left(\Sps, (\psi^1 \cup \phi^{2}_{S} \cup \phi^{3}_{S}, \phi^{1}_{u} \cup \phi^{2}_{u}, \phi^{2}_{V \backslash \Sps}) \right)\notag\\
	\leq &f\left(\Sps, (\phi^2_{S}\cup \phi^{3}_{S},\phi^2_{u}\cup \phi^{3}_{u}, \phi^{2}_{V \backslash \Sps})  \right) - f\left(\Sps, (\phi^2_{S}\cup \phi^{3}_{S},\phi^2_{u}, \phi^{2}_{V \backslash \Sps})\right) .\label{eq:adaptive-gap-upper12}
	\end{align}
	Consider any node $v \in \Gamma(\Sps,(\psi^1 \cup \phi^{2}_{S} \cup \phi^{3}_{S}, \phi^{1}_{u} \cup \phi^{2}_{u} \cup \phi^{3}_{u}, \phi^{2}_{V \backslash \Sps}) ) \backslash \Gamma(\Sps, (\psi^1 \cup \phi^{2}_{S} \cup \phi^{3}_{S}, \phi^{1}_{u} \cup \phi^{2}_{u}, \phi^{2}_{V \backslash \Sps}) )$, we have the following observations: 
	(1) Node $v$ cannot be reached from any node in set $\Sps$
	under the realization $(\psi^1 \cup \phi^{2}_{S} \cup \phi^{3}_{S}, \phi^{1}_{u} \cup \phi^{2}_{u}, \phi^{2}_{V \backslash \Sps})$; and
	(2) node $v$ can be reached via a simple path $P$ originated from node $u$
	under the realization $(\psi^1 \cup \phi^{2}_{S} \cup \phi^{3}_{S}, \phi^{1}_{u} \cup \phi^{2}_{u} \cup \phi^{3}_{u}, \phi^{2}_{V \backslash \Sps})$, and $P$ does not contain any node in $S$
	and any edge in $\phi^{1}_{u} \cup \phi^{2}_{u}$.

	Now, we prove that $v \in \Gamma(\Sps, (\phi^2_{S}\cup \phi^{3}_{S},\phi^2_{u}\cup \phi^{3}_{u}, \phi^{2}_{V \backslash \Sps})  ) \backslash \Gamma(\Sps, (\phi^2_{S}\cup \phi^{3}_{S},\phi^2_{u}, \phi^{2}_{V \backslash \Sps}))$.  
	Since path $P$ does not contain any node in $S$ and any edge in $\phi^{1}_{u}$, we know that path $P$ also exists under realization $ (\phi^2_{S}\cup \phi^{3}_{S},\phi^2_{u}\cup \phi^{3}_{u}, \phi^{2}_{V \backslash \Sps}) $, i.e., node $v$ can be reached from node $u$ under realization $ (\phi^2_{S}\cup \phi^{3}_{S},\phi^2_{u}\cup \phi^{3}_{u}, \phi^{2}_{V \backslash \Sps}) $. 
	Moreover, we know that the realization $(\phi^2_{S}\cup \phi^{3}_{S},\phi^2_{u}, \phi^{2}_{V \backslash \Sps})$
	has less live edges than the realization $(\psi^1 \cup \phi^{2}_{S} \cup \phi^{3}_{S}, \phi^{1}_{u} \cup \phi^{2}_{u}, \phi^{2}_{V \backslash \Sps})$, 
	thus node $v$ cannot be reached from the set $\Sps$ under realization $(\phi^2_{S}\cup \phi^{3}_{S},\phi^2_{u}, \phi^{2}_{V \backslash \Sps})$. Thus we can conclude that $v \in \Gamma(\Sps, (\phi^2_{S}\cup \phi^{3}_{S},\phi^2_{u}\cup \phi^{3}_{u}, \phi^{2}_{V \backslash \Sps})  ) \backslash \Gamma(\Sps, (\phi^2_{S}\cup \phi^{3}_{S},\phi^2_{u}, \phi^{2}_{V \backslash \Sps}))$, this leads to Eq.~\eqref{eq:adaptive-gap-upper12} and concludes the proof
	of Inequality \eqref{eq:adaptive-gap-upper6}.
\end{proof}

{\lemHybridvsNonAdaptive*}
\begin{proof}
	Again, for ease of notation, we set $S = \dom(\psi^1)$ and $\Sps = \dom(\psi^1) \cup \{u\}$, then we have
	\begin{align}
	&\Delta_{f^3}(u \mid \psi^1) = 
	\E_{\Phi^1, \Phi^2, \Phi^3 \sim \cP}\left[f\left(\Sps,(\Phi^{1}_{S} \cup \Phi^{2}_{S} \cup \Phi^{3}_{S}, \Phi^{1}_{u} \cup \Phi^{2}_{u} \cup \Phi^{3}_{u}, \Phi^{1}_{V \backslash \Sps}) \right)
	\right. \nonumber \\
	& \qquad \left. - f\left(S, (\Phi^{1}_{S} \cup \Phi^{2}_{S} \cup \Phi^{3}_{S}, \Phi^{1}_{V \backslash S}) \right) \mid \Phi^1 \sim \psi^1\right]\nonumber\\
	& =\E_{\Phi^1, \Phi^2, \Phi^3 \sim \cP}\left[f\left(\Sps,(\Phi^{1}_{S} \cup \Phi^{2}_{S} \cup \Phi^{3}_{S}, \Phi^{1}_{u} \cup \Phi^{2}_{u} \cup \Phi^{3}_{u} , \Phi^{1}_{V \backslash \Sps}) \right)\right.\nonumber\\
	&\qquad - \left.f\left(\Sps, (\Phi^{1}_{S} \cup \Phi^{2}_{S} \cup \Phi^{3}_{S}, \Phi^{1}_{u} \cup \Phi^{2}_{u}, \Phi^{1}_{V \backslash \Sps}) \right) \mid \Phi^1 \sim \psi^1\right]\nonumber\\
	&+\E_{\Phi^1, \Phi^2, \Phi^3 \sim \cP}\left[f\left(\Sps,(\Phi^{1}_{S} \cup \Phi^{2}_{S} \cup \Phi^{3}_{S}, \Phi^{1}_{u} \cup \Phi^{2}_{u}, \Phi^{1}_{V \backslash \Sps}) \right) - f\left(S, (\Phi^{1}_{S} \cup \Phi^{2}_{S} \cup \Phi^{3}_{S}, \Phi^{1}_{V \backslash S}) \right) \mid \Phi^1 \sim \psi^1\right]\nonumber\\	
	&\leq \Delta_{f^2}(u | S) + \Delta_{f^2}(u | S) 
	= 2\Delta_{f^2}(u | \dom(\psi^1)).\label{eq:adaptive-gap-upper7}
	\end{align}
	The inequality above is a direct consequence of Lemmas \ref{lemma:adaptive-gap-myopic1}.
\end{proof}

{\leminfluenceSpreadHighOrder*}
\begin{proof}
	We have
	\begin{align}
	\sigma^{t}(S) &= \E_{\Phi^1, \cdots, \Phi^t \sim \cP}\left[f^{t}(S,\Phi^1, \cdots, \Phi^t)\right]
	=\E_{\Phi^1, \cdots, \Phi^t \sim \cP}\left[f\left(S, (\cup_{i \in [t]} \Phi^{i}_{S}, \Phi^{1}_{V\backslash S})\right)\right]\notag\\
	&=\E_{\Phi^{1}_{V\backslash S} \sim \cP_{V\backslash S}}\left[\E_{\Phi^{1}_{S}, \cdots, \Phi^{t}_{S} \sim \cP_{S}}\left[f\left(S, (\cup_{i \in [t]} \Phi^{i}_{S}, \Phi^{1}_{V\backslash S})\right)\right]\right]. \label{eq:adaptive-gap-upper20}
	\end{align}
	We want to show that for any fixed $\phi^1_{V\setminus S}$, 
	\begin{align}
	& \E_{\Phi^{1}_{S}, \cdots, \Phi^{t}_{S} \sim \cP_{S}}\left[f\left(S, (\cup_{i \in [t]} \Phi^{i}_{S}, \phi^{1}_{V\backslash S})\right)\right]
	\leq \sum_{i \in [t]}\E_{\Phi^{i}_{S}}\left[f\left(S, (\Phi^{i}_{S}, \phi^{1}_{V\backslash S})\right)\right]\label{eq:adaptive-gap-upper21}.
	\end{align}
	Once Eq.\eqref{eq:adaptive-gap-upper21} is shown, we can combine with Eq.\eqref{eq:adaptive-gap-upper20}
	to obtain
	\begin{align*}
	\sigma^{t}(S) 
	& \le \E_{\Phi^{1}_{V\backslash S}  \sim \cP }
	\left[ \sum_{i \in [t]}\E_{\Phi^{i}_{S}}\left[f\left(S, (\Phi^{i}_{S}, \Phi^{1}_{V\backslash S})\right)\right]\right] \nonumber \\
	& = \sum_{i \in [t]} \E_{\Phi^{1}_{V\backslash S} \sim \cP }
	\left[\E_{\Phi^{i}_{S}}\left[f\left(S, (\Phi^{i}_{S}, \Phi^{1}_{V\backslash S})\right)\right]\right]
	\\
	& = \sum_{i \in [t]} \E_{\Phi^{1} \sim \cP} \left[f(S, \Phi^1) \right] = t\cdot \sigma(S).
	\end{align*}
	Thus the lemma holds. Now we prove Inequality \eqref{eq:adaptive-gap-upper21}.
	To do so, we fix partial realizations $\phi^{1}_{S}, \cdots, \phi^{t}_{S}$. 
	If node $v \in \Gamma(S,\cup_{i \in [t]} \phi^{i}_{S}, \phi^{1}_{V\backslash S}))$, then we conclude that under the realization $(\cup_{i \in [t]} \phi^{i}_{S}, \phi^{1}_{V\backslash S})$, node $v$ can be reached via a path $P$ originated from some node $u \in S$, and only the starting node of $P$
	is in $S$ and all remaining nodes in $P$ are not from $S$. 
	Suppose in path $P$, the edge leaving node $u$ is contained in edge set $\phi^{i}_{u}$ for some $i \in [t]$. Then we conclude that  node $v \in \Gamma(S, (\phi^{i}_{S}, \phi^{1}_{V\backslash S}))$, since the path $P$ exists under the realization $(\phi^{i}_{S}, \phi^{1}_{V\backslash S})$. 
	This shows that $\Gamma(S, (\cup_{i \in [t]}\phi^{i}_{S}, \phi^{1}_{V\backslash S})) \subseteq
	\cup_{i \in [t]} \Gamma(S, (\phi^{i}_{S}, \phi^{1}_{V\backslash S})) $, which
	is sufficient to prove Inequality \eqref{eq:adaptive-gap-upper21}.
\end{proof}

\section{Missing Proof of Section \ref{sec:adaptive-gap-lower-bound},
	Adaptivity Lower Bound}
\label{sec:missing-proof}

{\thmAdaptiveGapLower*}
\begin{proof}
	Consider the following construction for the influence graph: the influence graph $G = (L, R, E, p)$ is a bipartite graph with $|L| = \binom{m^3}{m^2}$ and $|R| = m^3$. All edges $(u, v) \in E$ are directed from the left part $L$ to the right part $R$, associated with probability $1/m$. More specifically, for any subset $X \subseteq R$ with size $m^2$, there is a node $u_{X} \in L$ such that the outgoing edges of $u_X$ are exactly $(u_X, v)$ for every $v \in X$. 
	Thus the out-degree of every vertex in $L$ is $m^2$. 
	
	
	We first describe the main idea of the proof.
	The budget for the IM problem is $m^2$, i.e., we are allowed to select no more than $m^2$ seeds, and we would consider $m$ to be a very large number here. 
	Intuitively, the expected number of nodes in $R$ that is reachable for a single node $u\in L$ is $m^2 \cdot (1/m) = m$, and the influence spread is concentrated on its expected value for large $m$. 
	In an adaptive solution, we could always make the expected marginal gain for the node we select equals the expected influence spread of a single node in $L$, by selecting nodes in $L$ such that none of its out-neighbors has been reached so far, unless there are too few nodes in $R$ that are not reachable.
	Since $m^2 \cdot m = m^3$, the seeds we select would reach almost all but except $o(m^3)$ nodes in $R$, thus the influence spread of the adaptive policy is roughly $m^3$. 
	While for a non-adaptive policy, it can select at most $m^2$ nodes from $L$ and for each node in $R$, on average, it is connected with at most $m^2 \cdot m^2 / m^3= m$ seeds in $L$, we can easily prove that it is indeed the best allocation of seeds in $L$, and the expected probability for nodes in $R$ to be reached is $1 - (1 - 1/m)^m \approx 1 - 1/e$. Moreover, since we are allowed to select no more than $m^2$ seeds in $R$ and they would not reach any other node, the contribution of this part is negligible. Thus the expected influence spread for the optimal non-adaptive solution would not exceed $(1 - 1/e)m^3$ and the adaptivity gap is $e/(e-1)$ on this graph.
	
	The following two claims would make the above intuition formal. 
	\begin{claim}
		\label{claim:myopic-adap-opt-exp}
		For any $\epsilon > 0$, when $m$ is large enough, we have $\OPT_A(G, m^2) \geq (1 - \epsilon)m^3$.
	\end{claim}
\begin{proof}
	For any fixed $\epsilon > 0$, we would take $m$ such that $m \geq 48/\epsilon^{2}\log m$. 
	Consider the following adaptive policy $\pi$, which only selects nodes from the left part $L$. 
	Moreover, for every node $u \in L$ selected by $\pi$, at the time of selection, none of $u$'s 
	out-neighbors in $R$ has been reached yet from nodes selected by $\pi$ so far (this condition can be verified by an adaptive policy with myopic feedback). 
	When there does not exist such node or the size of the seed set already equals to the budget, $\pi$ would stop.
	For $i \in \{1, \cdots, (1 - \epsilon/2)m^2\}$, let $\Evt_{i}$ denote the event that after selecting the $i$-th seed in $L$, the marginal gain of the influence spread is between $[(1 - \epsilon/2)m + 1, (1 + \epsilon/2)m + 1]$. 
	We would give a lower bound on the conditional probability $\Pr[\Evt_i \mid \Evt_1, \cdots, \Evt_{i - 1}]$. Under the condition $\cup_{j = 1}^{i - 1}\Evt_j$, the current influence spread on the right 
	part $R$ 
	 is less than $(1 + \epsilon/2)m \cdot (1 - \epsilon/2)m^2 = (1 - \epsilon^2/4)m^3 < m^3 - m^2$, thus policy $\pi$ would not stop by now. 
	Thus the marginal gain is the summation of $m^2$ independent binomial variables with mean $m$. By the Chernoff bound we have 
	\begin{align}
	\label{eq:missing-proof-eq1}
	Pr[\Evt_i \mid \Evt_1, \cdots, \Evt_{i - 1}] \geq 1 - exp(-\epsilon^2 m/12) \geq 1 - \frac{1}{m^3}.
	\end{align}
	Consequently,
	\begin{align}
	\label{eq:missing-proof-eq2}
	Pr[\cup_{i = 1}^{t}\Evt_i] = \Pi_{i = 1}^{t}\Pr[\Evt_i \mid \Evt_1, \cdots, \Evt_{i - 1}] \geq (1 - \frac{1}{m^3})^{m^2} \geq 1 - \frac{1}{m^3}\cdot m^2 = 1 - \frac{1}{m}.
	\end{align}	
	Thus the expected influence is greater than $(1 - \frac{1}{m}) \cdot (1-\epsilon/2)m \cdot(1 - \epsilon/2)m^2 \geq (1 - \epsilon)m^3$.
\end{proof}
	
	\begin{claim}
		\label{claim:myopic-nonadap-opt-exp}
		$\OPT_{N}(G, m^2) \leq (1 - (1 - 1/m)^m)m^3 + 2m^2$.
	\end{claim}
\begin{proof}
	Let $S_L$ (resp. $S_R$) denote the seed set selected by the optimal non-adaptive policy from the left part $L$ (resp. right part $R$). 
	For any node $u_i \in R$ where $i \in [m^3]$, let $x_i$ denote the number of $u_i$'s in-neighbors in the seed set $S_L$. 
	Since the out-degree for each node in $S_L$ is $m^2$, we have $\sum_{i \in [m^3]} x_i \leq |S_L| \cdot m^2$ and the average number of in-neighbors is at most $|S_L| \cdot m^2/m^3 = |S_L|/m$. 
	Furthermore, we can calculate the influence spread of $S_L$, 
	\begin{align}
	\sigma(S_L) &= |S_L| + \sum_{i \in [m^3]} \Pr[u_i \text{ is reachable}]\notag\\
	&=|S_L| + \sum_{i \in [m^3]}\left(1 - \left(1 - \frac{1}{m}\right)^{x_i}\right)\notag\\
	&\leq |S_L| + m^3 \cdot \left(1 - \left(1 - \frac{1}{m}\right)^{|S_L|/m} \right)\notag\\
	&\leq m^2 + m^3\cdot \left(1 - \left(1 - \frac{1}{m}\right)^{m} \right).\label{eq:missing-proof-equation1}
	\end{align} 
	The first inequality holds because function $g(x)=(1 - (1 - \frac{1}{m})^{x})$ is concave. 
	The last inequality holds because $|S_{L}| \leq m^2$.
	Now we have
	\begin{align}
	\OPT_{N}(G, m^2) &= \max\limits_{\substack{S_{L}\subseteq L, S_{R}\subseteq R,\\ |S_L| + |S_R| \leq m^2}}\sigma(S_L \cup S_R) 
	\leq \max\limits_{\substack{S_{L}\subseteq L,\\ |S_L|\leq m^2}}\sigma(S_L) + \max\limits_{\substack{S_{R}\subseteq L,\\ |S_R|\leq m^2}}\sigma(S_R)\notag\\ 
	&\leq m^2 + m^3\cdot \left(1 - \left(1 - \frac{1}{m}\right)^{m}\right) + m^2 = \left(1 - \left(1 - \frac{1}{m}\right)^{m}\right)\cdot m^3 + 2m^2.
	\end{align}
	This concludes the proof.
\end{proof}

	Combining Claims \ref{claim:myopic-adap-opt-exp} and \ref{claim:myopic-nonadap-opt-exp}, we 
	can conclude that for any $\epsilon > 0$, there exists large enough $m$ such that
	$\OPT_A(G, m^2)  / \OPT_{N}(G, m^2) \ge e/(e-1) - \epsilon$.
	Letting $\epsilon \rightarrow 0$, we obtain the theorem.
\end{proof}

\section{Missing Proofs in Section~\ref{sec:constant-competitive-algo}}

For the proofs in this section, 
	let $\G_N(G, k)$ (resp. $\G_A(G,k)$) denote the influence spread for the non-adaptive greedy algorithm
	(resp. adaptive influence spread for the adaptive greedy algorithm), 
	on the influence graph $G$ with a budget $k$.

The proof of Theorem~\ref{theorem:greedy&adaptive-greedy-constant} is complete once
	we prove the following lemma.
\begin{lemma}
	\label{lemma:adaptive-greedy-vs-non-adaptive-optimal}
	Adaptive greedy is $(1 - 1/e)$ approximate to the optimal non-adaptive policy.
\end{lemma}
\begin{proof}
	For a fixed influence graph $G$, let $S$ ($|S| = k$) denote the seed set selected by the optimal non-adaptive algorithm, where $s_i$ denotes the $i^{th}$ element in set $S$. 
	We use $\cA$ to denote adaptive greedy and for any $t \in \{0, 1, \cdots, k\}$, we use $U(t)$ to denote the expected adaptive influence spread of nodes selected by $\cA$ in the first $i$ rounds, i.e., 
	\begin{align}
	\label{eq:adaptive-greedy-eq1}
	U(t) := \E_{\Phi \sim \cP}\left[f\left(V(\cA, \Phi)_{:t}, \Phi \right)\right],
	\end{align}
	From the above definition, we can see that $U(0) = 0$ and $U(k) = \sigma(\cA)$. By Lemma \ref{lem:marginal}, we have
	\begin{align}
	\label{eq:adaptive-greedy-marginal}
	U(t) = \sum_{i = 0}^{t-1}\E_{s \sim \cP_{i}^{\cA}}\left[\Delta_{f}\left(\cA(\psi_s) \mid \psi_s\right)\right].
	\end{align}
	Now, for any $t \in \{0,1, \cdots, k - 1\}$
	\begin{align}
	U(t + 1) - U(t) &= \E_{s \sim \cP_{t}^{\cA}}\left[\Delta_f\left(\cA(\psi_s) \mid \psi_s\right)\right] \notag\\
	&\geq \frac{1}{k} \sum_{i = 1}^{k}\E_{s \sim \cP_{t}^{\cA}}\left[\Delta_f\left(s_i \mid \psi_s\right)\right]\notag \\
	&= \frac{1}{k}\sum_{i=1}^{k} \E_{s \sim \cP_{t}^{\cA}}\left[\E_{\Phi\sim \cP}\left[f\left(\dom(\psi_s) \cup \{s_i\}, \Phi\right) - f\left(\dom(\psi_s), \Phi\right) | \Phi \sim \psi_s\right]\right]\notag\\
	&= \frac{1}{k} \E_{s \sim \cP_{t}^{\cA}}\left[\E_{\Phi\sim \cP}\left[\sum_{i=1}^{k}\left(f\left(\dom(\psi_s) \cup \{s_i\}, \Phi\right) - f\left(\dom(\psi_s), \Phi\right)\right) | \Phi \sim \psi_s\right]\right]\notag\\
	&\geq \frac{1}{k}\E_{s \sim \cP_{t}^{\cA}}\left[\E_{\Phi \sim \cP}\left[f(\dom(\psi_s) \cup S, \Phi) - f(\dom(\psi_s), \Phi) | \Phi \sim \psi_s\right]\right]\notag\\
	&\geq \frac{1}{k}\E_{s \sim \cP_{t}^{\cA}}\left[\E_{\Phi\sim\cP}\left[f(S, \Phi) - f(\dom(\psi_s), \Phi) | \Phi \sim \psi_s\right]\right]\notag\\
	&=\frac{1}{k}\left(\fa(S) - U(t)\right).  \label{eq:adaptive-greedy-margin}
	\end{align} 
	The first inequality holds since adaptive greedy $\cA$ chooses the node that maximizes the expected marginal gain, i.e., for any partial realization $\psi$, $\Delta_f(\cA(\psi) \mid \psi ) \geq \Delta_f(s_i \mid \psi )$ for any $i \in [k]$.
	The second inequality is because the influence utility function $f(\cdot, \Phi)$ is submodular 
	under a fixed realization $\Phi$. 
	The third inequality holds because the influence utility function $f(\cdot, \Phi)$ is monotone 
	under a fixed realization $\Phi$. 
	The last equality utilizes the law of total expectation.
	
	Now via standard argument, Eq.~\eqref{eq:adaptive-greedy-margin} implies that
	\begin{align}
	\G_{A}(G, k) &= U(k) \geq \left(1 - \left(1 - \frac{1}{k}\right)^{k}\right)\sigma(S) = \left(1 - \left(1 - \frac{1}{k}\right)^{k}\right)\OPT_N(G, k)\notag\\
	&\geq \left(1 - \frac{1}{e} \right) \cdot \OPT_N(G, k).
	\label{eq:adaptive-greedy-eq2}
	\end{align}
	This concludes the proof.
\end{proof}

%
%

We now prove Theorem~\ref{theorem:bad-example-greedy}.
We first present a example showing that the non-adaptive greedy achieves at most
	$\frac{e^2 + 1}{(e + 1)^2}$ approximation ratio.

\begin{lemma}
	\label{lem:greedyBadExample}
	Non-adaptive greedy algorithm has ratio at most $\frac{e^2 + 1}{(e + 1)^2}$ with respect to the
	optimal adaptive solution, in the IC model with myopic feedback.
\end{lemma}
\begin{proof}
	Consider the following influence graph $G(V, E, p)$, where $V = V_1 \bigcup V_2 \bigcup V_3$, 
	{$|V_1| = d -1$}, $|V_2| = d$ and $|V_3| = 2d$. 
	We would use $v^{i}_{j}$ to denote the $j^{th}$ node in $V_{i}$. Nodes in $V_1$ and $V_2$ have unit weight while nodes in $V_3$ have weight $w$. 
	Note that we could achieve the weight of $w$ by simply replacing each node
		with a chain of $w$ nodes with edge probability $1$, so that as long as the head of the
		chain is activated, the whole chain is activated. 
	There are directed edges from $V_1$ to $V_2$ and from $V_2$ to $V_3$. 
	More specifically, for any $j \in [d], l \in [d-1]$, there is a direct edge from the node $v^{1}_{l}$ to the node $v^{2}_{j}$, associated with probability $1/d$. 
	The node $v^{2}_{j}$ is connected to node $v^{3}_{2j - 1}$ and $v^{3}_{2j}$, with probability $e/(e + 1)$. 
	The budget $k = \frac{e + 3}{e + 1}d$. We first consider the optimal adaptive solution and we observe that the optimal adaptive strategy can reach almost all nodes in $V_3$.
	\begin{claim}
		\label{thm:adapt-opt-perform}
		For any $\epsilon > 0$, if we set $d \geq 2\log (2/\epsilon)/\epsilon^2$, then we have $\OPT_{A}(G, k) \geq (1 - \epsilon)\cdot 2dw$
	\end{claim}
	\begin{proof}
		Consider the following adaptive strategy: we first select all $d$ nodes in $V_2$ and observe which nodes in $V_3$ have not yet been reached, this can be done with myopic feedback. 
		We would then use the left budget to select nodes in $V_3$ that have not been reached. 
		Let $X_j = \I\{v^{3}_{j} \mbox{ not activated by seed nodes in $V_2$}\}$ for $j \in [2d]$,
			where $\I\{\}$ is the indicator function. 
		$X_j$'s are independent Bernoulli random variables with $\E[X_{j}] = \frac{1}{e + 1}$.
		Then by the Chernoff bound,
		\begin{align}
		\label{eq:chernoff-adaptive-opt}
		\Pr[X_1 + \cdots + X_{2d} > \frac{2}{e + 1}d + \epsilon d ]\leq e^{-\frac{\epsilon d \cdot \epsilon(e+1)/2}{3}} \leq  e^{-d\epsilon^2/2} \leq \frac{\varepsilon}{2}. 
		\end{align}
		Consequently, the expected number of nodes in $V_3$ that have not been activated 
			by seeds in $V_2$ is at most
		 $\frac{\epsilon}{2}\cdot 2d + (1-\frac{\epsilon}{2})\cdot (\frac{2}{e + 1}d + \epsilon d) 
		 \leq \frac{2}{e + 1}d + 2\epsilon d$.
		But the adaptive greedy algorithm still has a budget of $\frac{2}{e + 1}d$ to directly
		activate nodes in $V_3$, and thus the expected final number of non-activated nodes in
		$V_3$ is at most $2\epsilon d$.
		Thus we conclude the proof.
	\end{proof}
	
	Next, we consider the greedy algorithm and have the following conclusion.
	\begin{claim}
		\label{thm:greedy-perfrom}
		The non-adaptive greedy algorithm would first select all {$d - 1$} nodes in $V_1$, and then select $\frac{2}{e + 1}d + 1$ nodes in $V_2$. Consequently, we have that
		\begin{align}
		\label{bad-example-greedy-eq1}
		\G_N(G, k) = \left(d - 1\right)+ \left[\left(\frac{2}{e + 1}d + 1\right) + \left(1 - \left(1 - \frac{1}{d}\right)^{d - 1}\right)\cdot\left(\frac{e - 1}{e + 1}d -1 \right) \right]\cdot(1 + \frac{2e}{e + 1}w),
		\end{align}
		when $d, w \rightarrow \infty$, we know that $\frac{\G_N(G, k)}{dw}  \rightarrow \frac{2e^2 + 2}{(e + 1)^2}$.
	\end{claim}
	
	\begin{proof}
		We first prove that greedy would first select all $d - 1$ nodes in $V_1$.
		Consider that the greedy algorithm has already selected $j$ nodes in $V_1$ as seeds,
			with $j=0,1,\ldots, d - 1$.
		Let $p_j$ denote the probability that a node in $V_2$ is activated in this case. 
		We know that $p_j = 1 - (1 - \frac{1}{d})^{j}$.  
		At this point, we know that the marginal gain for selecting the $(j+1)$-th node in $V_1$ is
		\begin{align}
		\label{bad-example-greedy-eq2}
		M_1 = 1 + d\cdot \frac{1}{d} (1 - p_j)\cdot (1 + \frac{2e}{e + 1}w) = 1 + (1 - p_j)(1 + \frac{2e}{e + 1}w).
		\end{align}
		In contrast, the marginal gain for selecting the first node in $V_2$ as a seed is 
		\begin{align}
		\label{bad-example-greedy-eq3}
		M_2 = (1 - p_j)(1 + \frac{2e}{e + 1}w),
		\end{align}
		and the marginal gain for selecting the first node in $V_3$ as a seed is
		\begin{align}
		\label{bad-example-greedy-eq4}
		M_3 = p_j(1 - \frac{e}{e + 1})w + (1 - p_j)w = \left(p_j \cdot \frac{1}{e + 1} + (1 - p_j)\right)w.
		\end{align}
		Therefore $M_1 > M_2$.
		Comparing $M_1$ with $M_3$, we use the fact that for all $j < d$, $p_j \le 1 - 1/e$, 
		and thus
%
		\begin{align}
		M_1 - M_3 &=1 + (1 - p_j)(1 + \frac{2e}{e + 1}w) - \left(p_j \cdot \frac{1}{e + 1} + (1 - p_j)\right)w\notag\\
		&> (1 - p_j)\frac{2e}{e + 1}w -  \left(p_j \cdot \frac{1}{e + 1} + (1 - p_j)\right)w\notag\\
		&=\left(\frac{e - 1}{e + 1} - \frac{e}{e + 1}p_j\right)w\notag\\
		&\geq 0.	\label{bad-example-greedy-eq5}
		\end{align}
		Thus we conclude that greedy would select all $(d - 1)$ nodes in $V_1$ first. 
		Afterwards, we compare the marginal gain of selecting a node in $V_2$ versus selecting
		a node in $V_3$.
		Notice that if we select a node in $V_3$, we would definitely not select a node whose
			in-neighbor in $V_2$ is already selected as a seed, because it only decreases the marginal.
		Therefore, the marginal gains of selecting a node in $V_2$ or a node in $V_3$ are still
			given us $M_2$ and $M_3$.
		Thus, the difference of marginal gain is
		\begin{align}
		M_2 - M_3 &= (1 - p_{d-1})(1 + \frac{2e}{e + 1}w) - \left(p_{d-1} \cdot \frac{1}{e + 1} + (1 - p_{d - 1})\right)w\notag\\
		&> (1 - p_{d-1})\frac{2e}{e + 1}w - \left(p_{d-1} \cdot \frac{1}{e + 1} + (1 - p_{d - 1})\right)w\notag\\
		&=\left(\frac{e - 1}{e + 1} - \frac{e}{e + 1}p_{d-1}\right)w\notag\\
		&=\left(\frac{e - 1}{e + 1} - \frac{e}{e + 1}\left(1 - \left(1 -\frac{1}{d}\right)^{d-1}\right)\right)w\notag\\
		&\geq 0.	\label{bad-example-greedy-eq6}
		\end{align} 
		Thus the marginal gain for selecting nodes in $V_2$ is greater than nodes in $V_3$ and greedy would select $\frac{2}{e + 1}d + 1$ nodes in $V_2$. 	All in all, the expected utility for greedy is

		
		\begin{align}
		\G_N(G, k) = \left(d - 1\right)+ \left[\left(\frac{2}{e + 1}d + 1\right) + \left(1 - \left(1 - \frac{1}{d}\right)^{d - 1}\right)\cdot\left(\frac{e - 1}{e + 1}d -1 \right) \right]\cdot(1 + \frac{2e}{e + 1}w).\label{bad-example-greedy-eq7}
		\end{align}
		and when $d, w \rightarrow \infty$, we know that 
			$\frac{\G_N(G, k)}{dw} \rightarrow \frac{2e^2 + 2}{(e + 1)^2}$.
	\end{proof}
	
	Combining Claim \ref{thm:greedy-perfrom} and Claim \ref{thm:adapt-opt-perform}, we conclude that when $d, w \rightarrow \infty$, 
	\begin{align}
	\label{bad-example-greedy-eq8}
	\frac{\G_{N}(G, k)}{\OPT_{A}(G)} \rightarrow \frac{e^2 + 1}{(e + 1)^2} \approx 0.606.
	\end{align}
\end{proof}

We then assert that the approximation ratio of adaptive greedy is no better than greedy.
\begin{lemma}
	\label{claim:greedy-vs-adaptive-greedy}
	The approximation ratio for the non-adaptive greedy algorithm is no worse than the adaptive greedy algorithm, over all graphs.
\end{lemma}
\begin{proof}
	Fix an influence graph $G(V, E, p)$, and any $k \in [n]$.
	We use $c$ to denote the approximation ratio of greedy, i.e., 
	\[
	c = \frac{\G_N(G, k)}{\OPT_A(G, k)}.
	\]
	We construct a family of graph $G(w)$ such that the approximation ratio for adaptive greedy is approaching to $c$ when $w \rightarrow \infty$.
	The influence graph $G(w)$ consists of two parts, $G_1$ and $G_2$. 
	The graph $G_1$ has same nodes as $G$, but it does not contain any edges, while the graph $G_2$ is exactly the same as $G$, except that the weight for each node is multiplied by a factor of $w$. 
	Notice that we can always assign integral weights $w$ to a node by connecting it to a directed chain of length $w - 1$. 
	For any node $v\in G_1$, $v$ has exactly one outgoing edge, connecting to the corresponding node in $G_2$, the edge will be live with probability 1.
	
	Now, consider adaptive greedy on $G(w)$ with the same budget.
	Our first observation is that adaptive greedy will never choose nodes from $G_2$. 
	This is because if the corresponding node in $G_1$ has not been chosen, the marginal gain of choosing the node in $G_1$ is always larger by 1, and if it has already been chosen, the marginal gain to choose the node in $G_2$ is 0. 
	Consequently, the adaptive greedy algorithm would always choose nodes in $G_1$.
	However, because myopic feedback only provides one step feedback after seed selection, 
	selecting a node in $G_1$ would only provide the activation of its corresponding node in $G_2$
		as the feedback, but this is already known for sure, and thus we do not get any
		useful feedback under myopic feedback model on this graph.
	Therefore, the adaptive greedy algorithm in this case behaves exactly the same as the 
		non-adaptive greedy algorithm on the influence graph $G$,
		 and the performance for adaptive greedy is
	\begin{align}
	\label{eq:adapt-greedy-perform}
	\G_{A}(G(w), k) = w \cdot \G_{N}(G, k) + k \leq (w + 1) \cdot \G_{N}(G, k) .
	\end{align}
	Consider the optimal adaptive policy, a feasible adaptive policy is to ignore nodes in graph $G_1$ and perform the optimal adaptive policy on graph $G_2$, we have
	\begin{align}
	\label{eq:adapt-opt-perform}
	\OPT_A(G(w), k) \geq \OPT_{A}(G_2(w), k) = w\cdot \OPT_{A}(G, k).
	\end{align}
	By Eq.~\eqref{eq:adapt-greedy-perform} and Eq.~\eqref{eq:adapt-opt-perform}, the approximation ratio of adaptive greedy can be bounded as
	\begin{align}
	\label{eq:adapt-greedy-approx}
	\frac{\G_{A}(G(t), k)}{\OPT_A(G(t), k)} \leq \frac{(w + 1) \cdot \G_{N}(G, k)) }{w\cdot \OPT_A(G, k)} =  \frac{w + 1}{w}\cdot c \rightarrow c, \mbox{ when } w\rightarrow \infty.
	\end{align}
	This concludes the proof.
\end{proof}

\end{document}